\documentclass[11pt]{amsart}

\pdfoutput = 1

\usepackage[english]{babel}
\usepackage[autostyle]{csquotes}
\usepackage{amsmath,amssymb,amsthm}
\usepackage{amsfonts}
\usepackage{amsxtra}

\usepackage{array,dcolumn}
\usepackage{graphicx}
\usepackage[hyperfootnotes=true,
        pdffitwindow=true,
        plainpages=false,
        pdfpagelabels=true,
        pdfpagemode=UseOutlines,
        pdfpagelayout=SinglePage,
                hyperindex,]{hyperref}

\usepackage{lmodern}
\usepackage[T1]{fontenc}
\usepackage[utf8]{inputenc}
\usepackage[activate={true,nocompatibility},spacing,kerning]{microtype}
\usepackage{bm}
\usepackage{bbm}
\usepackage[sc,osf]{mathpazo}
\microtypecontext{spacing=nonfrench}

  \calclayout

\newcommand{\mathsym}[1]{{}}
\newcommand{\unicode}[1]{{}}

\theoremstyle{plain}
\newtheorem{theorem}{Theorem}

\newtheorem{proposition}[theorem]{Proposition}

\theoremstyle{definition}

\theoremstyle{remark}
\newtheorem{remark}[theorem]{Remark}

\numberwithin{equation}{section}

\numberwithin{theorem}{section}

\numberwithin{figure}{section}

\begin{document}


\title[Rank $1$ perturbations in random matrix
theory]{Rank $1$ perturbations  in random matrix
theory --- a review of exact results}
\author{Peter J. Forrester}
\address{School of Mathematics and Statistics, 
ARC Centre of Excellence for Mathematical \& Statistical Frontiers,
University of Melbourne, Victoria 3010, Australia}
\email{pjforr@unimelb.edu.au}

\date{\today}


\begin{abstract}
A number of random matrix ensembles permitting exact determination of their eigenvalue
and eigenvector statistics maintain this property under a rank $1$ perturbation. Considered
in this review are the additive rank $1$ perturbation of the Hermitian Gaussian ensembles,
the multiplicative rank $1$ perturbation of the Wishart ensembles, and 
rank $1$ perturbations of Hermitian and unitary matrices giving rise to a two-dimensional support for the eigenvalues.
The focus throughout is on exact formulas, which are typically the result of various integrable structures.
The simplest is that of a determinantal point process, with others relating to partial differential equations
implied by a  formulation in terms of certain random tridiagonal matrices. Attention is also given to
eigenvector overlaps in the setting of a rank $1$ perturbation.
\end{abstract}


\maketitle


\section{Introduction}\label{S1}
In the mid 90's I was  possession of early edition of Wolfram's {\it The Mathematic Book}.
The introductory gallery section contained the command, up to the accuracy of my
memory
\begin{multline}\label{1.1}
{\tt ComplexListPlot[
 Eigenvalues[RandomVariate[UniformDistribution[], {100, 100}]], }\\
 {\tt PlotRange -> All]}
\end{multline}
and the accompanying output reproduced below in Figure \ref{F1}. Thus a random $100 \times 100$
matrix, with each entry identically and independently formed drawn from the uniform distribution on
$[0,1]$, was formed, the eigenvalues were calculated, and these were plotted in the complex plane.

Now a random variable $u[0,1]$, uniformly chosen from the interval $[0,1]$ can be decomposed
${1 \over 2} + u[-{1 \over 2},{1 \over 2}]$, where $u[-{1 \over 2},{1 \over 2}]$ is a uniform random variable with support $[-{1 \over 2},{1 \over 2}]$.
Hence the random matrix being formed in (\ref{1.1}) can be written
\begin{equation}\label{1.2}
X + {1 \over 2} \mathbb I_{N \times N} = X + {N \over 2} \, \hat{\mathbf 1}_N   \hat{\mathbf 1}_N^T.
\end{equation}
Here  $\mathbf  I_{N \times N} $ denotes the $N \times N$ matrix with all entries $1$,
and $\hat{\mathbf 1}$ is the $N \times 1$ unit column vector with all entries equal to
$1/\sqrt{N}$.
Also, $X$ is the random matrix with all entries identically and independently
distributed as uniform random variables on $[-1/2,1/2]$, and thus having mean
zero and standard deviation $1/\sqrt{12}$.

\begin{figure*}
\centering
\includegraphics[width=0.95\textwidth]{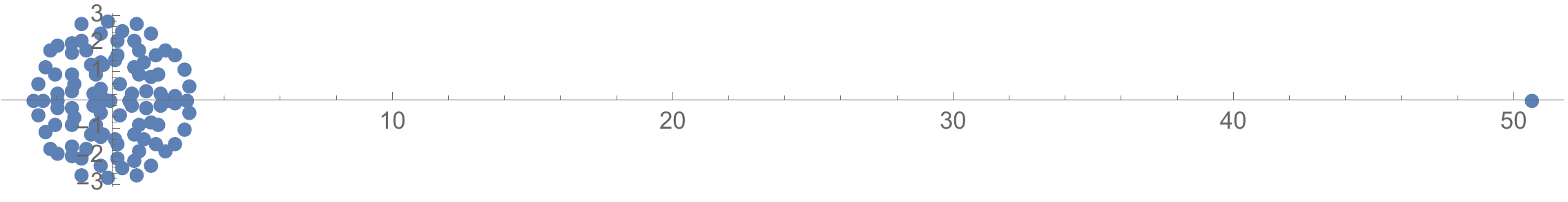}
\caption{Output of the Mathematica code (\ref{1.1}).}
\label{F1}
\end{figure*}

A celebrated result in random matrix theory, known as the circular law (see \cite{BC12} for a review),
tells us that an $N \times N$ random matrix with all entries identically and independently
distributed with mean zero and standard deviation $1/\sqrt{N}$ has, for large $N$, its eigenvalues uniformly supported on the 
unit disk in the complex plane. Moreover, the probability that an eigenvalue is of order unity outside of this
disk fall off exponentially fast in $N$. This suggests we write the second expression in (\ref{1.2})
as
\begin{equation}\label{1.3}
\sqrt{N \over 12}Y, \qquad Y : =  \Big (  \tilde{X} + \sqrt{3N} \, \hat{\mathbf 1}_N   \hat{\mathbf 1}_N^T \Big ),
\end{equation}
where in the definition of $Y$ the $N \times N$ random matrix $\tilde{X}$ obeys the conditions required for the applicability of
the circular law. The spectrum of $Y$ is obtained from the random matrix (\ref{1.2}) by the simple
scaling of dividing by $\sqrt{N \over 12}$.

In the setting of Figure \ref{F1}, $\sqrt{N \over 12} \approx 2.9$. Dividing  the scale of both axes by this value
it is observed  that all but one of the eigenvalues form a disk of radius $1$, which to the naked eye has a uniform
density. However, there is also an outlier eigenvalue, appearing on the real axis with value close to
$\sqrt{3N} |_{N = 100} \approx 17.3$. The latter is precisely the scalar of the perturbation of $\tilde{X}$ in (\ref{1.3})
by the addition of a scalar times the rank 1 matrix formed by the unit vector $\hat{\mathbf 1}_N$.
These are general features. Thus for $\alpha > 1$, $\hat{\mathbf v}$ a unit column vector and 
$ \tilde{X}$ obeying the conditions required of the circular law, we have that for large $N$
the random matrix
\begin{equation}\label{1.4}
  \tilde{X} + \alpha   \hat{\mathbf v}  \hat{\mathbf v}^T 
\end{equation}  
conforms to the circular law, with a single outlier eigenvalue on the real axis at $x = \alpha$,
as established by Tao \cite{Ta12}. For applications of the matrix structure (1.4) --- specifically
the averaged absolute value of the corresponding determinant --- to the stability and resilience of
large complex systems, see \cite{BKK16, FFI21}.

An even earlier numerical experiment relating to outlier eigenvalues in random matrix spectra
was carried out by Porter in the earlier 1960's,
 as cited in \cite{La64}.  The random matrices $\{ H \}$ say in this experiment were real symmetric, with Gaussian
 entries, distribution  N$[\mu, 1]$ on the diagonal and distribution N$[\mu, 1/\sqrt{2}]$ off the diagonal.
 The joint distribution of all the independent entries gives for that the probability density of the matrices
 $H$ are proportional to $e^{- {\rm Tr} \, (H - \mu \mathbb I_{N \times N})^2/2}$.  We can write
  \begin{equation}\label{1.5}
  H =  G + { \mu N \over 2}    \hat{\mathbf 1}_N   \hat{\mathbf 1}_N^T, 
  \end{equation} 
  where  $G$ has a probability density proportional to $e^{- {\rm Tr} \, G^2/2}$. 
  The latter is invariant under the mapping $G \mapsto RGR^T$
  for $R$ real orthogonal, which gives rise to the name of
the random matrices $G$ as the Gaussian orthogonal ensemble (GOE); the random matrices
$H$ are examples of particular shifted GOE matrices. Let us scale $H$ by multiplying by a factor of
$1/\sqrt{2N}$, and let us replace $\mu N/2 \sqrt{2N}$ by $\alpha$. This replaces (\ref{1.5}) by
  \begin{equation}\label{1.6}
  \tilde{G} + \alpha   \hat{\mathbf 1}_N   \hat{\mathbf 1}_N^T,
  \end{equation} 
  where $\tilde{G}$ is $1/\sqrt{2N}$ times a GOE matrix. A classical result in random matrix theory
(see \cite{PS11})  tells us that the eigenvalues of the latter are
to leading order supported on the interval $[-1,1]$, with
normalised density
 \begin{equation}\label{1.6a}
\rho^{\rm W}(x) = {2 \over \sqrt{\pi}} (1 - x^2)^{1/2},
\end{equation}
known as the Wigner semi-circle.
As in the case of the circular law, for the scaled GOE matrices $ \tilde{G}$
the probability that an eigenvalue
is of order unity outside of this
interval falls off exponentially fast in $N$ \cite{BDG01,Fo12}.
The effect observed in Porter's simulations 
--- Figure \ref{F3} gives an example produced using
modern software ---
is that for $N$ large and $\alpha > 1/2$ there is a single outlier eigenvalue
 located at the value 
  \begin{equation}\label{1.6b}
 \alpha + 1/(4 \alpha),
 \end{equation}
 with the Wigner semi-circle otherwise remaining unchanged; see also Section \ref{S2.1}.
To leading order in $\alpha$ for $\alpha$ large this was first explained theoretically in
Lang  \cite{La64}, although as reviewed in Section \ref{S2.1} this exact value is now well
 understood theoretically.
 
 \begin{figure*}
\centering
\includegraphics[width=0.95\textwidth]{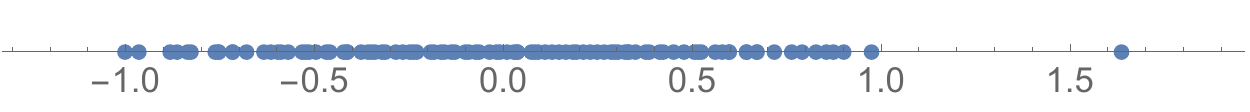}
\caption{Eigenvalues of a sample random matrix (\ref{1.6}) with
$N = 100$ and $\alpha = 3/2$. Note that the theoretical location
of the separated eigenvalue is then $5/3$. }
\label{F3}
\end{figure*}
 
The above examples have involved additive rank $1$ perturbations of a random matrix. 
Closely related is a sequence of additive  rank $1$ perturbations --- to be referred to as
(additive) rank $1$ updates --- which in fact can be used to define a discrete random evolution.
Consider for example the sequence of random matrices $\{W_n\}_{n=1,2,\dots}$
defined by
 \begin{equation}\label{1.7}
 W_n = \sum_{j=1}^n  {\mathbf v}_j  {\mathbf v}^T_j, 
 \end{equation}
 where each $\mathbf v_j$ is an independent $N \times 1$ standard Gaussian column vector,
 and calculate their eigenvalues. This is easy
 to simulate. An example is given in Figure \ref{F4}. Here there is no weighting of the
 rank $1$ matrices and thus no eigenvalue separation
 phenomenon, but visible is another feature of an Hermitian rank $1$ perturbation of 
 an Hermitian matrix, namely that of an interlacing of eigenvalues.

The purpose of this survey is to give an account of exact formulas, typically driven by underlying integrability, associated with 
rank $1$ perturbations 
 in random matrix theory. The topic of \S \ref{S2} is the additive rank $1$ structure (\ref{1.6}). Considered in this section are the
 derivation of the formula (\ref{1.6b}) for the location of the outlier, and generalisations; an explicit formula for the  eigenvalue PDF
 from the viewpoints of an underlying tridiagonal matrix, and from a matrix integral related to Dyson Brownian motion; a characterisation of the
 distribution of the largest eigenvalue in the critical regime in terms of a partial differential equation; a $\beta$ generalisation of the latter and
 its solution for $\beta = 2,4$ in terms of Painlev\'e transcendents; and overlap properties of the eigenvector corresponding to the
 largest eigenvalue.  The topic of \S \ref{S3} is a multiplicative rank $1$ perturbation of a complex Wishart matrix.
 It is shown how such a multiplicative perturbation can be recast as an additive rank $1$ perturbation, allowing the theory
 of subsection \ref{S2.1} to be used to determine the criteria and location of an outlier. Two derivations of the explicit formula for the
 eigenvalue PDF are given, one involving the HCIZ matrix integral, and the other computing first the joint distribution of the eigenvalues
 of the matrix involved in the equivalent additive rank $1$ perturbation, and the perturbed matrix. The eigenvalues of the perturbed complex Wishart
 matrix under consideration form a determinantal point process, and the explicit form of the correlation kernel for the soft edge critical regime
 is revised in subsection \ref{S3.3}. The final subsection relates to the hard edge critical regime for general $\beta > 0$.
 Rank $1$ perturbations of Hermitian and unitary matrices giving rise to a two-dimensional support for the eigenvalues is the
 topic of \S \ref{S4}. First considered in this section is an additive anti-symmetric perturbation for the GUE. For a scaling close to the origin
 of the real axis, this gives rise to a determinantal point process for the eigenvalues, with a simple functional form for the kernel. Next
 a multiplicative sub-unitary rank $1$ perturbation of Haar distributed unitary matrices is considered. In a scaling near the unit circle in the
 complex plane, the eigenvalue point process is identical to that of the previous subsection. However in the bulk of interior of the unit
 circle the eigenvalue point process is no longer determinantal, and in fact relates to the zeros of a certain class of random Laurent
 series, related to the limiting Kac polynomial. In the final subsection overlaps between the left and right eigenvectors of the setting of
 subsection \ref{S4.1} are considered.

\begin{figure*}
\centering
\includegraphics[width=0.65\textwidth]{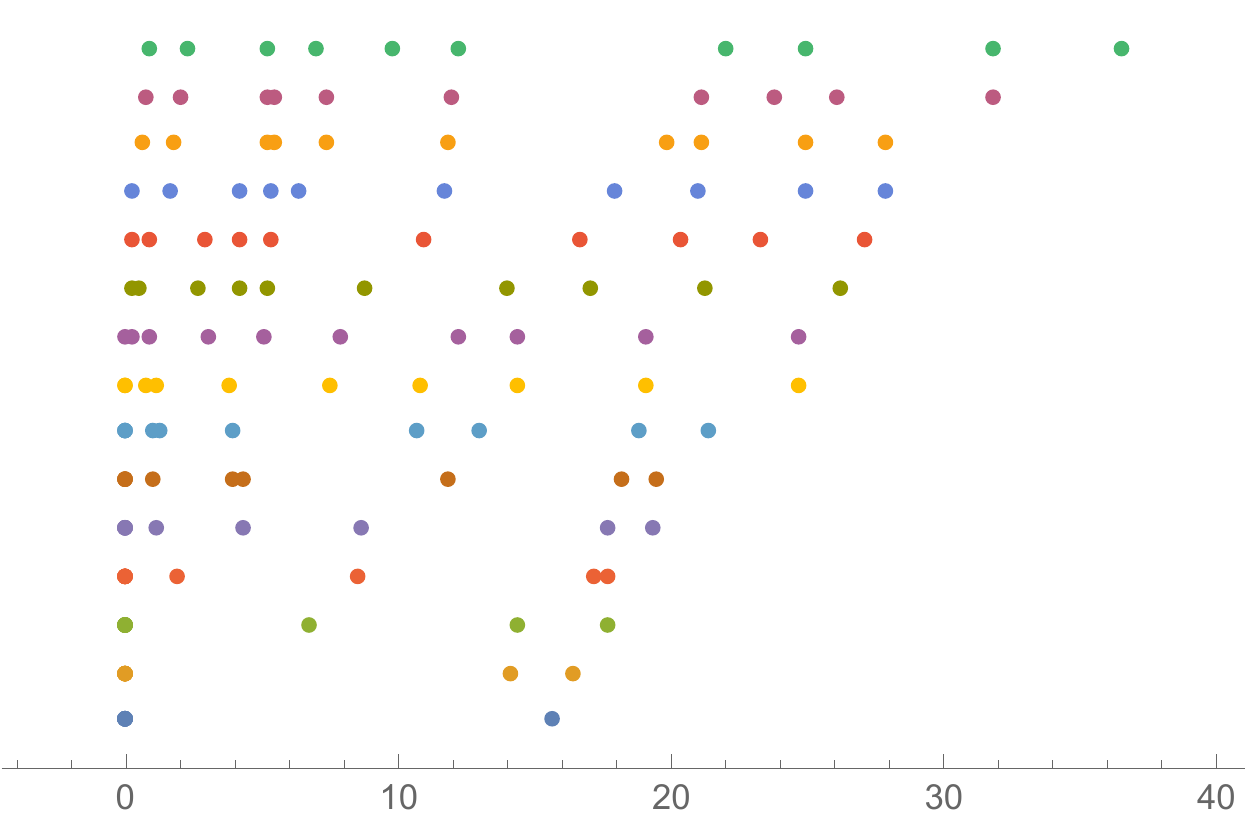}
\caption{Distinct eigenvalues of a sequence of random matrices (\ref{1.7}) with
$n=1,2,\dots,15$ reading from bottom to top and $N = 10$.}
\label{F4}
\end{figure*}

\section{An additive rank 1 perturbation for the GOE}\label{S2}
\subsection{Location of the separated eigenvalue for $\alpha > 1/2$}\label{S2.1}
The formula (\ref{1.6b}) for the separated eigenvalue was given by Jones,
Kosterlitz and Thouless in 1978
 \cite{JKT78}, and independently in more general setting by  Furedi and Komlos
in 1981 \cite{FK81}.  We follow the derivation of these works.
 
 \begin{proposition}\label{P2.1}
 Consider the particular shifted scaled GOE matrix (\ref{1.6}), suppose $\alpha > 0$ and take the limit $N \to \infty$.
 For $\alpha  \le1/2$ all eigenvalues are supported on $[-1,1]$ and have density given by the Wigner semi-circle (\ref{1.6a}).
 For $\alpha > 1/2$ all but one of the eigenvalues are supported on $[-1,1]$ and have density given by the Wigner semi-circle,
 with the separated eigenvalue located at the value (\ref{1.6b}).
 \end{proposition}
 
 \begin{proof}
 The characteristic equation determining the eigenvalues of (\ref{1.6}) is
  \begin{equation}\label{1.8a}
 0 = \det ( \mu \mathbb I_{N \times N} - \tilde{G} - \alpha   \hat{\mathbf 1}_N   \hat{\mathbf 1}_N^T) =
 \det \Big ( \mu   \mathbb I_{N \times N} - \tilde{G} - \alpha   \hat{\mathbf e}_N^{(1)}   \hat{\mathbf e}_N^{(1) \, T} \Big ),
  \end{equation}
 where $ \hat{\mathbf e}_N^{(1)} $ denotes the unit column vector in $\mathbb R^N$, $(1,0,\dots,0)$.
 Here the second equality follows from the fact, noted below (\ref{1.5}), that the distribution of $\tilde{G}$
 is invariant under conjugation by a real orthogonal matrix, allowing the rank $1$ matrix
 $ \hat{\mathbf 1}_N   \hat{\mathbf 1}_N^T$ to be replaced by its diagonal matrix of eigenvalues
 ${\rm diag} \, (1,0,\dots,0)$, then writing the latter as $ \hat{\mathbf e}_N^{(1)}   \hat{\mathbf e}_N^{(1) \, T}$.
 
 Next introduce the spectral decomposition $( \lambda \mathbb I_{N \times N} - \tilde{G})^{-1} =
 U  ( \lambda \mathbb I_{N \times N} - \Lambda)^{-1} U^T$, where $\Lambda = {\rm diag} \, (\mu_1,\dots, \mu_N)$
 is the diagonal matrix of eigenvalues of $\tilde{G}$, and $U = [ \mathbf{u}_k ]_{k=1}^N$ is the corresponding
 real orthogonal matrix of eigenvectors. With this substituted in (\ref{1.8a}), the decisive step in the argument is to apply the
 general determinant identity (see e.g.~\cite[Corollary 2.1]{Ou81})
   \begin{equation}\label{1.8b}
   \det ( \mathbb I_{N \times N} - A_{N \times M} B_{M \times N}) =  \det ( \mathbb I_{M \times M} - B_{M \times N}  A_{N \times M} ) 
   \end{equation}
to the second determinant in (\ref{1.8a}), thereby reducing it to the scalar
\begin{equation}\label{1.8c}
1 - {\alpha } \sum_{j=1}^N { (u_1^{(j)})^2 \over \lambda -  \mu_j },
 \end{equation}
 where $u_1^{(j)}$ denotes the $j$-th component of $\mathbf u_1$. Since the first factor in (\ref{1.8a})
 has zeros at the poles of (\ref{1.8c}), the eigenvalues as determined by (\ref{1.6}) are seen to
 be given by the  zeros of (\ref{1.8c}) as a function of $\mu$.
 
 For $\{ \mu_j \}_{j=1}^N$ distinct and ordered
 \begin{equation}\label{1.8d}
 \mu_N < \mu_{N-1} < \cdots < \mu_2 < \mu_1,
  \end{equation}
  a sketch of the graph of (\ref{1.8c}) under the assumption $\alpha > 0$ shows that its zeros,
  $\{ \lambda_j \}_{j=1}^N$ say, ordered from biggest to smallest are also distinct, and interlace with
  (\ref{1.8d}) according to
  \begin{equation}\label{1.8e} 
   \mu_N < \lambda_N <
   \mu_{N-1} <  \lambda_{N-1} < \cdots < \mu_2  < \lambda_2 < \mu_1 < \lambda_1.
   \end{equation}
   For the scaled GOE matrix $\tilde{G}$ we know the eigenvalues  $\{ \mu_j \}_{j=1}^N$ for
   large $N$ concentrate on $[-1,1]$, having density (\ref{1.6a}), and there are no outliers.
   It remains then to determine the location of the largest zero of (\ref{1.8c}) in this setting.
   
   Thus we average (\ref{1.8c}) over the eigenvalues of the scaled GOE, and the components
   of the first eigenvector. We know that the eigenvectors are independent of the eigenvalues, and
 are distributed uniformly on the unit sphere in $\mathbb R^N$
 (see \cite[Eq.~(1.11) in relation to the former point and Exercises 1.2 q.2 in relation to the latter]{Fo10}), telling us that each
 $(u_1^{(j)})^2 $ can be replaced by $1/N$. The sum in the remaining quantity
   \begin{equation}\label{1.8f} 
  1 - {\alpha \over N} \sum_{j=1}^N { 1 \over \lambda -  \mu_j } 
  \end{equation}
  is a linear statistic, which when averaged over the eigenvalues can be written in terms
  of the eigenvalue density according to $\int_{\mathbb R} (\rho_{(1),N}(\lambda)/(\mu - \lambda)) \, d \lambda$.
  For large $N$ we have that the normalised density $(1/N) \rho_{(1),N}(\lambda)$ tends to the 
  Wigner semi-circle (\ref{1.6a}), implying that condition for a zero of (\ref{1.8c}) reduces to
    \begin{equation}\label{1.8g}  
    0 =  1 -  {2 \alpha \over \sqrt{\pi} }  \int_{-1}^1 { (1 -  \lambda^2)^{1/2} \over \mu -  \lambda} \, d \lambda.
   \end{equation}
   
   We have the integral evaluation (see e.g.~\cite[Exercises 1.6 q.2(ii)]{Fo10})
  \begin{equation}\label{1.8h}      
   {2  \over \sqrt{\pi} }  \int_{-1}^1 { (1 -  \lambda^2)^{1/2} \over \mu -  \lambda} \, d \lambda =
 \left \{  \begin{array}{ll} 2 \mu, & |\mu| \le 1 \\
 2 \mu (1 - (1 - 1/ \mu^2)^{1/2} ), & | \mu| > 1. \end{array} \right.
 \end{equation}
 We see by substituting (\ref{1.8h}) in (\ref{1.8g})
 that the latter admits a solution with $\mu > 1$ only if $\alpha > 1/2$, in which case
 it is given by (\ref{1.6b}).
 \end{proof}
 
 \begin{remark}\label{R2.2} 1.~Consider the additive rank $1$ perturbation (\ref{1.6}),  with $\tilde{G}$ belonging
 to a general ensemble of matrices with limiting normalised density $\rho^{\tilde{G}}(x)$, supported on $[a,b]$;
 references addressing this setting include \cite{CDF09, BR11, BL16, Ro16}.  Define the Stieltjes transform of the latter
   \begin{equation}\label{G1}
   G^{(\tilde{G})}(y) = \int_a^b {\rho^{\tilde{G}}(x) \over y - x} \, dx.
  \end{equation}  
  According to the above proof, eigenvalue separation occurs whenever the equation
   \begin{equation}\label{G2}
   {1 \over \alpha} =  G^{(\tilde{G})}(y)
  \end{equation}  
  admits a solution with $y = y^* > b$. This will always be the case for $\alpha$ large enough
  since for $y$ large, it follows from the definition (\ref{G1}) and the normalisation of
  $\tilde{G}(x)$ that $ G^{(\tilde{G})}(y) \sim 1/y$, telling us that $y^* \sim \alpha$. This
  latter conclusion is in keeping with the early finding of Lang \cite{La64}. \\
  2.~Generalising (\ref{1.6}) to include several additive multiple rank $1$ perturbations 
  of strengths $\alpha_i$, involving
  linearly independent unit vectors, the eigenvalue separation equation (\ref{G2}) applies to each
 such perturbation separately, with the other perturbations ignored; see the cited references from
 point 1 above. \\
 3.~Suppose the real symmetric matrix $\tilde{G}$ in (\ref{1.6}) is replaced by $\tilde{G}_1 + t \tilde{G}_2$,
 where $\tilde{G}_1$ is real symmetric, $t \ge 0$ is a real parameter, and $ \tilde{G}_2$ is real anti-symmetric.
 If the independent entries of each $\tilde{G}_i$ are identically and independently distributed with mean
 zero and unit variance, after scaling by dividing by $\sqrt{N}$ the eigenvalue density satisfies the
 elliptical law. The location of outliers due to low rank perturbations in this setting, which interpolates between Hermitian matrices
 satisfying the Wigner semi-circle law, and non-Hermitian matrices satisfying the circular law
 (recall the third paragraph of the Introduction) have been studied in \cite{OR14}.
 \end{remark}
 
 \subsection{Joint eigenvalue probability density function}\label{S2.2}
 According to (\ref{1.8a}) the eigenvalues of the shifted scaled GOE matrix (\ref{1.6})
 are the same as for $\tilde{G}$, but with $\alpha \sqrt{2N}$ added to the entry in the top
 right corner. Thus, with eig$\,A$ denoting the eigenvalues of $A$, we have
 $$
 {\rm eig} \, ( \tilde{G}   - \alpha   \hat{\mathbf 1}_N   \hat{\mathbf 1}_N^T) = {1 \over \sqrt{2N}}
  {\rm eig} \, \Big ( G - \alpha \sqrt{2N} \, {\rm diag} \, (1,0,\dots,0) \Big ),
  $$
  where $G$ is a GOE matrix (no scaling). It was observed by Trotter \cite{Tr84} that a sequence
  of Householder reflector transforms, $R$ say, can be applied symmetrically to $G$ to reduce it to
  the tridiagonal form
  $$
  T = R G R^T = A_0 + A_1 + A_1^T,
  $$
  where
   \begin{equation}\label{A0}
   A_0 = {\rm diag} \, ( {\rm N}\,[0,1],\dots,{\rm N}\,[0,1]), \quad
   A_1 = {\rm diag}^+ \, ( \tilde{\chi}_{N-1}, \tilde{\chi}_{N-2},\dots, \tilde{\chi}_1),
 \end{equation}
 with $\tilde{\chi}_k$ denoting the square root of the gamma distribution $\Gamma[k/2,1]$.    In $A_1$
 the notation diag${}^+$ refers to a matrix with all entries equal to zero, except those on the diagonal
 immediately above the main diagonal, which take the values as listed. As is evident from the form
 of $A_0$, this
 transformation leaves the diagonal entries unchanged, and so \cite{BV16}
\begin{equation}\label{A2}   
{\rm eig} \, ( \tilde{G}   - \alpha   \hat{\mathbf 1}_N   \hat{\mathbf 1}_N^T) = {1 \over \sqrt{2N}}
  {\rm eig} \, \Big ( T - \alpha \sqrt{2N}  \, {\rm diag} \, (1,0,\dots,0) \Big ).
  \end{equation} 
  
  Following working introduced in the context of determining the eigenvalue PDF for a multiplicative
  rank perturbation of Wishart matrices \cite{Fo13} (see also Section \ref{S3.4} as well as \cite{Mo12,Wa12}),
  the random tridiagonal
  matrix in (\ref{A2}) can be used to determine the joint eigenvalue probability density function (PDF)
  for the shifted scaled GOE matrix (\ref{1.6}).

 \begin{proposition}\label{P2.3} 
Up to normalisation, the eigenvalue PDF of the random matrix (\ref{1.6}),
eigenvalues ordered $\lambda_1 > \cdots > \lambda_N$
is proportional to
\begin{equation}\label{A3}   
\prod_{j=1}^N e^{- N \lambda_j^2} \prod_{1 \le j < k \le N} ( \lambda_j - \lambda_k)
\int_{- \infty- i c}^{\infty - i c} e^{i t} \prod_{j=1}^N \Big ( i t - 2 \alpha N \lambda_j \Big )^{-1/2} \, dt,
\end{equation} 
where $c > 2 \alpha N \lambda_{1}$.
 \end{proposition}
 
 \begin{proof}
 The matrix in brackets on the RHS of (\ref{A2}) say, $\tilde{T}$ say, differs from $T$ in the distribution
 of the top left entry. Thus for $\tilde{T}$ this entry has distribution N$[\alpha \sqrt{2N},1]$ instead of
 N$[0,1]$ for $T$. 
 
 Now denote the entries of $\tilde{T}$ by writing $\tilde{T} = \tilde{A}_0 +
 \tilde{A}_1 +  \tilde{A}_1^T$, where
 $$
  \tilde{A}_0 = {\rm diag} \, (a_N,a_{N-1},\dots, a_1), \quad
    \tilde{A}_1 = {\rm diag}^+ \, (b_{N-1},b_{N-2},\dots, b_1).
    $$
It follows that up to normalisation the probability measure $P(\tilde{T})  (d \tilde{T})$    associated
with $\tilde{T}$ can be factored in terms of the probability measure $P({T})  (d {T})$    associated
with $T$ to be given by
$$
P({T})  (d {T}) e^{\alpha \sqrt{2N} a_N}.
$$
Next denote the eigenvalues of the tridiagonal matrix $\tilde{T}$ by $\{ \lambda_j \}_{j=1}^N$, and
denote the first component of the corresponding normalised eigenvector, which is required to be
positive, by $\{ q_j \}_{j=1}^N$. We know from the working in \cite{DE02}, or from \cite[\S 1.9.2]{Fo10},
that in terms of these variables and up to normalisation the probability measure $P({T})  (d {T})$ is proportional to
\begin{equation}\label{M1}
\prod_{j=1}^N e^{- \lambda_j^2/2} \prod_{1 \le j < k \le N}
( \lambda_j - \lambda_k ) \, \delta \Big ( \sum_{j=1}^N q_j^2 - 1 \Big )
(d \vec{\lambda})(d \vec{q}),
\end{equation} 
while
\begin{equation}\label{M2}
e^{\alpha \sqrt{2N} a_N} = e^{\alpha \sqrt{2N} \sum_{j=1}^N q_j^2 \lambda_j}.
\end{equation} 
Writing the Dirac delta function in (\ref{M1}) as a Fourier transform and supposing
temporarily that each $\lambda_j < 0$ allows the integral over $d \vec{\lambda}$ in
the product of (\ref{M2}) and (\ref{M1}) to be computed explicitly, showing
\begin{equation}\label{M2a}
\int_{(\mathbb R^+)^N}  \delta \Big ( \sum_{j=1}^N q_j^2 - 1 \Big )  e^{\alpha \sqrt{2N} \sum_{j=1}^N q_j^2 \lambda_j} \,  (d \vec{q}) 
\propto  \int_{-\infty}^\infty e^{it} \prod_{j=1}^N \Big ( i t - \alpha \sqrt{2N} \lambda_j \Big )^{-1/2} \, dt.
\end{equation} 
Deforming the contour to again be parallel to the  real axis but to pass through the imaginary axis
at a point $-ic$ with $c >\alpha \sqrt{2N} \lambda_1$ allows the assumption  $\lambda_j < 0$ to be removed.
Multiplying this modified integral 
by the eigenvalue dependent factors in (\ref{M1}) and scaling $\lambda_j \mapsto \sqrt{2N} \lambda_j$ to
account for the scaling in (\ref{A2}) gives (\ref{A3}).
\end{proof}
 
 \subsection{A matrix integral over the orthogonal group}\label{S2.3}
 There is a matrix integral over the orthogonal group associated with (\ref{A3}).
 This is based on the fact that for a real symmetric matrix $\tilde{H}$ the product of differentials
 of independent elements $(d \tilde{H})$ decomposes in terms of the eigenvalues $\{\lambda_j \}$ and
 matrix of eigenvectors $R$ of $\tilde{H}$ according to \cite[Eq.~(1.11)]{Fo10}
 $$
 (d \tilde{H}) = \prod_{1 \le j < k \le N} (\lambda_k - \lambda_j) d \vec{\lambda}
 (R^T d R),
 $$
 where $(R^T d R)$ corresponds to the Haar measure on the orthogonal group.
 In the case that the distribution on the matrices $\tilde{H}$ is proportional to
 $e^{- N {\rm Tr} (\tilde{H} - \alpha \hat{\mathbf 1}_N   \hat{\mathbf 1}_N^T)^2}$, it follows
 that the eigenvalue PDF is proportional to
 \begin{multline}\label{M3}
 \prod_{1 \le j < k \le N}
( \lambda_j - \lambda_k ) \int_{R \in O(N)} e^{-N {\rm Tr} (R \Lambda R^T  - \alpha 
 \hat{\mathbf 1}_N   \hat{\mathbf 1}_N^T )^2}  \, (R^T d R)   \\ \propto 
 \prod_{j=1}^N e^{-N \lambda_j^2 } \prod_{1 \le j < k \le N}
( \lambda_j - \lambda_k )  
\int_{R \in O(N)} e^{2 \alpha N {\rm Tr} (R \Lambda R^T  \hat{\mathbf 1}_N   \hat{\mathbf 1}_N^T)}
  \, (R^T d R),
  \end{multline}
  where $\Lambda$ is the diagonal matrix of eigenvalues.
  Comparing (\ref{M2}) with (\ref{A3}) implies an evaluation of the matrix integral in
  the second line of the former \cite{KTJ76,Fo13,BL16,PB20,MP21}.
  
   \begin{proposition}\label{P2.4} 
  We have
   \begin{multline}\label{M4}
  \int_{R \in O(N)} e^{2 \alpha N {\rm Tr} (R \Lambda R^T  \hat{\mathbf 1}_N   \hat{\mathbf 1}_N^T)}
  \, (R^T d R)  \propto 
  \int_{(\mathbb R^+)^N}  \delta \Big ( \sum_{j=1}^N q_j^2 - 1 \Big )  e^{2 \alpha N \sum_{j=1}^N q_j^2 \lambda_j} \,  (d \vec{q}) \\ \propto
\int_{-\infty-ic}^{\infty -ic} e^{it} \prod_{j=1}^N \Big ( i t - 2 \alpha N \lambda_j \Big )^{-1/2} \, dt ,
 \end{multline}
 where $c > 2 \alpha N \lambda_1$.
  \end{proposition}
  
 \begin{proof}
 We have already explained how the first matrix integral has an evaluation, up to proportionality, given by the final
 of these integrals. In relation to the second integral, we have used the fact that
 $$
 {\rm Tr} (R \Lambda R^T  \hat{\mathbf 1}_N   \hat{\mathbf 1}_N^T) =
 \vec{q}^T   \Lambda \vec{q} = \sum_{j=1}^N  q_j^2 \lambda_j.
 $$
 Here the first of these equalities follows from 
 the cyclic property of the trace and the fact that for $R$ a real orthogonal matrix chosen
 with Haar measure, $ \hat{\mathbf 1}_N^T R =  \vec{q}^T$, where $ \vec{q}$ is uniformly
 distributed on the unit sphere; see e.g.~\cite{DF17}.
 \end{proof}  
 
 \begin{remark} The second integral in (\ref{M4}) has the interpretation as the partition function of
 a spherical spin glass model \cite{KTJ76}. In the circumstance that $\{\lambda_j\}$ result
 the eigenvalues of a GOE matrix, it is shown in \cite{KTJ76,BL16} that the saddle point equation
 of the integrand in the evaluation given by the third integral in (\ref{M4}) relates to (\ref{1.8g}).
 This in turn implies that for large $N$ the model undergoes a phase transition as a function of
 $\alpha$.
 \end{remark}
 
 Let us rescale and generalise the (un-normalised) matrix distribution 
  $e^{- N {\rm Tr} (\tilde{H} - \alpha \hat{\mathbf 1}_N   \hat{\mathbf 1}_N^T)^2}$, by considering
\begin{equation}\label{R1}  
{1 \over t^{N(N+1)/2}} e^{- {\rm Tr} (\tilde{H} - \tilde{H}_0)^2/ 2t}.
\end{equation} 
Here the $t$-dependence of the normalisation has been made explicit for later purposes.
Changing variables to the eigenvalues and eigenvectors of $\tilde{H}$, analogous to 
(\ref{M3}), one has that the eigenvalue PDF of $\tilde{H}$,
$p_t(\lambda_1,\dots,\lambda_N)$, is proportional to 
\begin{equation}\label{R2}  
\prod_{1 \le j < k \le N} | \lambda_k - \lambda_j|
{1 \over t^{N(N+1)/2}} \int_{R \in O(N)} 
 e^{- {\rm Tr} (R \Lambda R^T - \Lambda_0)^2/ 2t} \, (R^T dR).
\end{equation} 
Implicit in the work of Dyson \cite{Dy62b}, on what is now known as Dyson Brownian
motion, is that $p_t$ satisfies a particular Fokker-Planck equation.

 \begin{proposition} 
 We have that $p_t$ satisfies the  Fokker-Planck equation 
 \begin{equation}\label{R3}  
 2 {\partial p_t \over \partial t} = \mathcal L p_t
\end{equation} 
where
 \begin{equation}\label{R4}  
 \mathcal L = \sum_{j=1}^N {\partial \over \partial \lambda_j} \bigg ( {\partial W \over \partial \lambda_j} +
 {\partial \over \partial \lambda_j} \bigg ), \qquad
 W = - \sum_{1 \le j < k \le N} \log | \lambda_k - \lambda_j|
 \end{equation} 
 and subject to the initial condition
  \begin{equation}\label{R4a}
   p_t(\lambda_1,\dots,\lambda_N) \Big |_{t=0} = \prod_{j=1}^N \delta (\lambda_j - \lambda_j^{(0)}).
   \end{equation} 
 \end{proposition}
 
 \begin{proof} (Sketch) With respect to the independent entries of $\tilde{H}$, the matrix distribution
 (\ref{R1}) factorises to be proportional to
  \begin{equation}\label{R4b} 
 \prod_{j=1}^N { e^{- (\tilde{H}_{jj}^2 - ( H_{jj}^{(0)})^2)/2t} \over \sqrt{t}}
 \prod_{j < k} { e^{- (\tilde{H}_{jk}^2 - ( H_{jk}^{(0)})^2)/t} \over \sqrt{t}},
  \end{equation} 
 where the $t$-dependence of the normalisation has been made explicit. Denoting this by $P_t$, the
 fact that the functional form corresponding to each independent element satisfies a one-dimensional
 heat equation implies that $P_t$ satisfies the multidimensional heat equation
  \begin{equation}\label{R5}  
 2 {\partial P_t \over \partial \tau} = \sum_{\mu} D_\mu {\partial^2 P_t \over \partial H_\mu^2}.
 \end{equation}  
 Here the label $\mu$ ranges over the label of the independent diagonal and upper triangular
 entries, while $D_\mu = 1$ for the diagonal entries and $D_\mu = 1/2$ for the off diagonal
 entries.

An essential idea from here, see e.g.~\cite[\S 11.1]{Fo10} for details, is to observe that as a function
of $\tilde{H} = R \Lambda R^T$, the PDF (\ref{R2}) must also satisfy
(\ref{R5}), provided we change variables in the latter. The change of variables can be carried out
using theory relating to the Laplacian associated with metric forms --- specifically the
RHS of (\ref{R5}) can be identified with the Laplacian operator on the space of
real symmetric matrices ---  and (\ref{R3}) results.
\end{proof}

\begin{remark}
As emphasised in \cite{Dy62b}), the Fokker-Planck equation specified by (\ref{R3}) and
(\ref{R4}) corresponds to a repelling $N$-particle system with a potential energy $W$,
executing overdamped Brownian motion in a fictitious viscous fluid with friction coefficient
$\gamma = 2$ at inverse temperature $\beta = 1$.
 \end{remark}

In the case of $\tilde{H}_0 =  \alpha \hat{\mathbf 1}_N   \hat{\mathbf 1}_N^T$, the initial condition
(\ref{R4a}) has
  \begin{equation}\label{R4c}
  \lambda_1^{(0)} = \alpha, \qquad \lambda_j^{(0)} = 0 \, \, (j=2,\dots,N).
 \end{equation}  
 With this initial condition the matrix distribution relating to (\ref{M3}), for which the density obeys the
 description of Proposition \ref{P2.1} when $N$ is large, results when $t = 1/(2N)$.
 The trajectories of the eigenvalues are easy to   simulate by choosing a value of $\delta t = 1/(2NM)$ for
 some $M \gg N$, forming a sequence of random real symmetric matrices $\{ \tilde{H}^{(j)} \}_{j=0,\dots,M}$
 by sampling the entries according to (\ref{R4b}) with $H^{(0)} =  \tilde{H}^{(j-1)}$ and $ t \mapsto \delta t$,
 calculating their eigenvalues and forming paths. An example is given in Figure \ref{F3a}.

 \begin{figure*}
\centering
\includegraphics[width=0.65\textwidth]{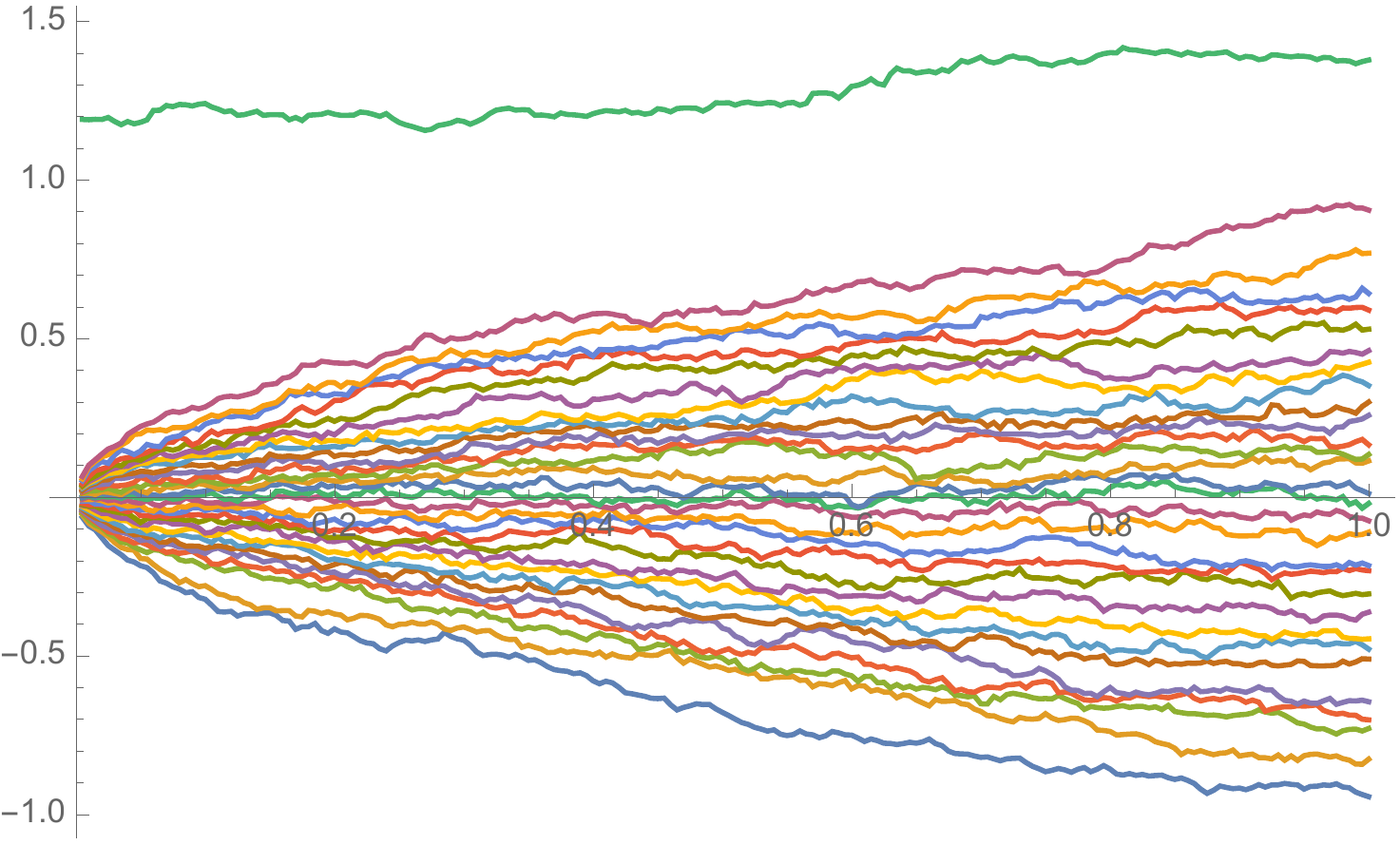}
\caption{Sample trajectories corresponding to the Dyson Brownian motion
underlying (\ref{M3}) with $N = 30$ and $\alpha = 1.2$.}
\label{F3a}
\end{figure*}

\subsection{The soft edge critical regime}\label{S2.4}
In \S \ref{S2.2} a random tridiagonal matrix was specified which has the same eigenvalue probability
density function as the rank $1$ perturbed scaled GOE matrix (\ref{1.6}). One application,
due to Bloemendal and Vir\'ag \cite{BV12}, of this reduction
has been to characterise the distribution of the largest eigenvalue $\lambda_1$ in a so-called critical regime. The latter
is specified by the large $N$ scalings
 \begin{equation}\label{C1}
 2 N^{2/3} (\lambda_1 - 1) \to x_1, \qquad N^{1/3} (1 - 2 \alpha) \to w,
  \end{equation} 
which have the feature that the distribution function then tends to a non-trivial limit dependent on $w$.
To gain insight, a parameter $\beta > 0$ referred to as the Dyson index can be introduced so that under consideration is the random tridiagonal
matrix
\begin{equation}\label{Tb}
T_\beta := {1 \over \sqrt{2 \beta N}} ( C_0 + C_1 + C_1^T),
 \end{equation} 
where
$$
C_0 = {\rm diag} \, \Big ( {\rm N}[\alpha \sqrt{2\beta N}, 1], {\rm N}[0,1],\dots, {\rm N}[0,1] \Big ), \qquad
C_1 = {\rm diag}^+ \, \Big ( \tilde{\chi}_{\beta (N - 1)},  \tilde{\chi}_{\beta (N - 2)}, \dots,  \tilde{\chi}_{\beta} \Big ).
$$
We can  verify the non-random limits
\begin{equation}\label{Tb+}
\lim_{\beta \to \infty} {1 \over \sqrt{\beta}} C_0 = {\rm diag} \, \Big ( \alpha \sqrt{2N}, 0,\dots, 0 \Big ), \quad
\lim_{\beta \to \infty} {1 \over \sqrt{\beta}}  C_1 = {\rm diag}^+ \, \bigg ( \sqrt{{N - 1 \over 2}},  \sqrt{N - 2 \over 2}, \dots,  \sqrt{1 \over 2} \bigg ).
 \end{equation} 
As a minor modification of working in \cite{DE02}, these facts can be used to show that
the large $\beta$ form of $T_\beta$ relates to a discretisation of the differential operator
 \begin{equation}\label{Tb-}
- {d^2 \over d x^2} + x,
 \end{equation} 
subject to the boundary conditions
 \begin{equation}\label{Tbb}
{\psi'(0) \over \psi(0)} = w , \qquad \psi(x) \to 0 \:\: {\rm as} \:\: x \to \infty.
 \end{equation}

 \begin{proposition} 
 Write $T_\beta = [t_{ij} ]_{i,j=1,\dots,N}$ and set
 $$
 D = {\rm diag} \, \Big ( (N/2)^{-(i-1)/2} \prod_{k=1}^{i-1} t_{k,k+1} \Big ).
 $$
 We have
 \begin{equation}\label{Tc}
 \lim_{\beta \to \infty} D T_\beta D^{-1} - \mathbb I_N = - {1 \over {2} N^{2/3}} \Big ( N^{2/3} \Delta_N +
 N^{-1/3} \tilde{J}_N \Big ),
 \end{equation} 
 where
  \begin{align*}
 \Delta_N & =  {\rm diag} \, (2,2,\dots,2) + {\rm diag}^+ \, (-1,-1,\dots,-1) +  {\rm diag}^- \, (-1,-1,\dots,-1), \\
  \tilde{J}_N & =  {\rm diag}^- \, (1,2,\dots,N-1).
  \end{align*}
 Furthermore the eigenvectors of (\ref{Tc}) are given by $D \mathbf x$ where $\mathbf x = [x_j]_{j=1,\dots,N}$ with
 $$
 x_{N-n} = \bigg ( {1 \over \sqrt{\pi} n! 2^n} \bigg )^{1/2} e^{-\lambda^2/2} H_n(\lambda),
 $$
  where $H_n(\lambda)$ denotes the Hermite polynomial of degree $n$, and $\lambda$ is required
 to be such that
 \begin{equation}\label{Tc+} 
 2 \alpha = {x_0\over x_1 }.
 \end{equation}  
 \end{proposition}
 
\begin{proof}
The normalised Hermite polynomials as a function of $\lambda$ and multiplied by $e^{-\lambda^2/2}$,
$\phi_n(\lambda)$ say, satisfy the recurrence
 \begin{equation}\label{X1}
 \lambda \phi_n(\lambda) = \sqrt{n/2} \phi_{n-1}(\lambda) + \sqrt{(n+1)/2} \phi_{n+1}(\lambda).
  \end{equation}  
  The structure of this recurrence is identical to that for the components of the eigen-equation for
  (\ref{Tb}), eigenvalue $\lambda$ and eigenvectors $[\phi_{N-1-j}(x) ]_{j=0}^{N-1}$, except in the first
  row. The latter requires
  $$
  \alpha \sqrt{2N} \phi_{N-1}(\lambda) + \sqrt{N-1 \over 2} \phi_{N-2}(\lambda) = \lambda \phi_{N-1}(x),
  $$
  which upon use of (\ref{X1}) with $n = N-1$ implies the restriction on $\lambda$ (\ref{Tc+}).
  \end{proof}
  
  Taking $N$ large in (\ref{Tc}) we recognise the right hand side as a discretisation of (\ref{Tb-}) with
  lattice spacing $N^{-1/3}$. With this latter value, upon rewriting (\ref{Tc+}) to read
  $$
  N^{1/3} (1 - 2 \alpha) = {1 \over x_1} \Big ( {x_1 - x_0 \over N^{-1/3}} \Big ),
  $$
  we see from (\ref{C1}) that the first of the boundary conditions in (\ref{Tbb}) results.
  
  As explained in \cite{DE02} and is readily verified,  taking $N$ large with 
  $\beta$ fixed, the appropriate modification of (\ref{Tc}) is that
   \begin{equation}\label{TcW} 
  D T_\beta D^{-1} -  \mathbb I_N = - {1 \over {2} N^{2/3}} \Big ( N^{2/3} \Delta_N +
 N^{-1/3} \tilde{J}_N    + {2 \over \sqrt{\beta}} W \Big ),
  \end{equation}  
 where $W$ is the bidiagonal random matrix
$$
 W = - {N^{1/6} \over \sqrt{2} }
\left [
\begin{array}{ccccc}
N[0,1] & & & &\\
b_{(N-1)\beta} & N[0,1] & & &\\
 & b_{(N-2)\beta} & N[0,1] & &\\
&  & \ddots & \ddots  &\\
& & & b_\beta & N[0,1] \end{array} \right ],
$$
with  $b_{(N-j) \beta} = (2 \tilde{\chi}_{(N-j)\beta}^2 - (N-j)\beta)/
\sqrt{2\beta N}$. A direct calculation shows
$b_{(N-j)\beta}$ has mean zero and variance $1 - j/N$, and so
each element of $W$ has mean zero, and to leading in $N$ 
for $j$ fixed
has standard deviation $N^{1/6}$.
This  is consistent with a discretisation, lattice spacing
$h = N^{-1/3}$, of a Brownian motion process 
which has mean zero and standard
deviation $\sqrt{h}$ over an interval $(x,x+h]$. Recalling that (\ref{Tc}) is a
discretisation of (\ref{Tb-}), these facts suggest that (\ref{TcW})
is a discretisation of the stochastic Airy operator
\begin{equation}\label{1.SA}
- {d^2 \over dx^2} + x + {2 \over \sqrt{\beta} } B'(x),
\end{equation}
where $B(x)$ defines a standard Brownian path. 

For $\alpha = 0$
the above reasoning was made rigorous
in the work \cite{RRV06} with the boundary condition $\psi(0) = 0$, and subsequently extended 
in \cite{BV12} to the case of nonzero $\alpha$
with the scaling as in (\ref{C1}) and boundary condition (\ref{Tbb}); in relation to the
latter see too \cite{LS19,LNR20}.
Note in particular that the ground state eigenvalue $\Lambda_0$ of (\ref{1.SA}) with
boundary conditions (\ref{Tbb}) corresponds to $-x_1$ as specified by the scaling of
the largest eigenvalue $\lambda_1$ for the tridiagonal matrix (\ref{Tb}). Transforming
(\ref{1.SA}) to a stochastic diffusion equation using a simple Ricatti change of independent
function $p(x) = {d \over dx} \psi(x)$ allows theory relating to Kolmogorov's backward
equation to be invoked. With $F_{\beta, w}(x)$ the cumulative distribution function
of $-\Lambda_0$, this implies the partial differential equation  \cite{BV12} 
\begin{equation}\label{1.K}
{\partial F \over \partial x} + {2 \over \beta} {\partial^2 F \over \partial w^2} + (x - w^2) {\partial F \over \partial w} = 0,
\end{equation}
subject to the boundary conditions that $F(x,w) \to 1$ as $x,w \to \infty$ simultaneously, and $F(x,w) \to 0$ as
$w \to -\infty$ with $x$ bounded above.

With regards to graphing the PDF of the largest eigenvalue in the critical regime,
use of (\ref{1.K}) is yet to demonstrate a numerical scheme with guaranteed accuracy. Instead, following
a suggestion in Edelman and Rao \cite{ER05} an accurate and efficient Monte Carlo procedure can be based on
(\ref{Tb}). First, it is argued that with respect to the largest eigenvalue and for $N$ large, truncating the $N \times N$
tridiagonal matrix to an $N_0 \times N_0$ tridiagonal  matrix with $N_0 \approx 10 N^{1/3}$
does not cause appreciable error. Moreover, the operations of storing a sparse (tridiagonal) matrix and computing
the largest eigenvalue, knowing that it is near $1$ are all highly efficient with modern software. Finally,
scaling the largest eigenvalue as required by (\ref{C1}), and repeating $M \gg 1$ times with $M$ large allows
for a histrogram approximating the PDF to be obtained for a given value of $w$.
An example is given in Figure \ref{F2.5}.

 \begin{figure*}
\centering
\includegraphics[width=0.65\textwidth]{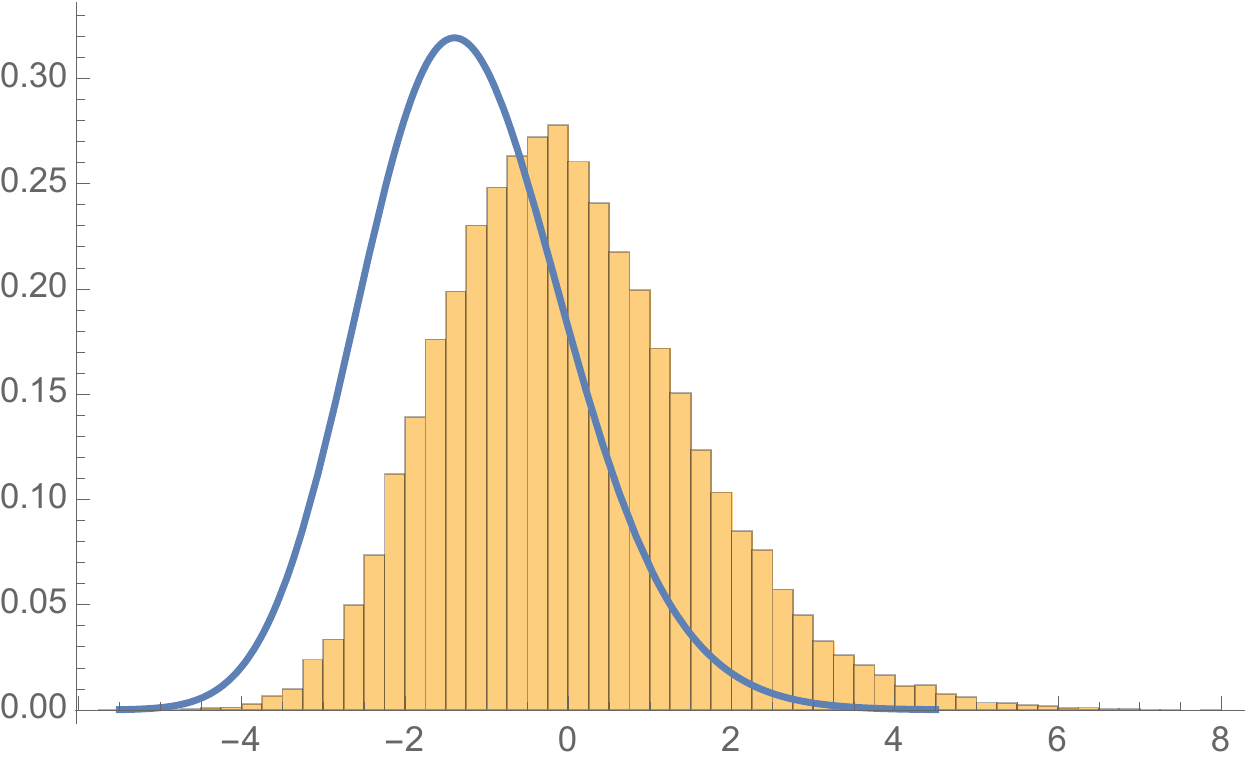}
\caption{Simulated histogram of the PDF for the largest
eigenvalue of (\ref{Tb}) with $\beta = 1$ in the scaled critical regime with
$N = 10^5$,  $\alpha = 1/2$ (equivalently $w=0$), and $M = 5*10^4$ repetitions. The solid
curve is the theoretical graph in the case $\alpha = 0$ (Tracy--Widom $\beta = 1$ distribution).}
\label{F2.5}
\end{figure*}

\begin{remark}
Combining a number of ideas, in particular that of a Pfaffian point process (see Remark \ref{R3.8} below)
and that of a Riemann-Hilbert characterisation of certain transcendents, a special function evaluation of
$F_{1, w}(x)$ has been given in \cite{Mo12}. However the complexity of this expression is such that it is
yet to be used for further analysis of properties $F_{1, w}(x)$, nor specifically for the determination of its numerical
values.
\end{remark}

\subsection{General $\beta > 0$}
In the previous section the Dyson index $\beta$ was introduced into the tridiagonal matrix formulation
of the rank $1$ perturbed GOE (\ref{1.6}) for convenience of motivating the ensuing working.
As known from the pioneering studies on random matrix theory undertaken by Dyson in the early
1960's, the first being \cite{Dy62}, there is a special significance in the three values $\beta = 1,2$
and 4. We know already that $\beta = 1$ corresponds to the GOE. The values $\beta = 2$ and $4$
correspond to the Gaussian unitary ensemble (GUE) and Gaussian symplectic ensemble (GSE),
for which the joint PDF of the elements of the corresponding Hermitian matrices $G$
is again $e^{-{\rm Tr} \, G^2/2}$. For the GUE the off-diagonal entries are complex, while for the GSE they
are quaternion; see \cite[\S 1.3]{Fo10}. With this choice of joint element PDF, upon scaling by $1/\sqrt{2N\beta}$
the eigenvalues to leading order are supported on the interval $[-1,1]$ and have limiting normalised
density given by the Wigner semi-circle (\ref{1.6a}) independent of $\beta$; see \cite{PS11}. From the
discussion of Remark \ref{R2.2}(1), the corresponding rank $1$ perturbation (\ref{1.6}) then has
eigenvalue separation properties as specified by Proposition \ref{P2.1}.

An extension of the tridiagonalisation procedure of Trotter \cite{Tr84}, as discussed in \S \ref{S2.2},
gives for each of the real, complex and quaternion cases of (\ref{1.6}) the tridiagonal matrix $T_\beta$
(\ref{Tb}). In fact the tridiagonal model allows Proposition \ref{P2.3} to be extended to general $\beta > 0$.

\begin{proposition}
Consider the rank $1$ perturbation  (\ref{1.6}) in the case of the GOE ($\beta = 1$),
GUE ($\beta = 2$) and GSE ($\beta = 4$). Extend this to general $\beta > 0$ according to the
tridiagonal matrix  (\ref{Tb}).
Up to normalisation, the eigenvalue PDF 
is proportional to
\begin{equation}\label{A3+}   
\prod_{j=1}^N e^{- \beta N \lambda_j^2} \prod_{1 \le j < k \le N} | \lambda_j - \lambda_k|^\beta
\int_{- \infty- i c}^{\infty - i c} e^{i t} \prod_{j=1}^N \Big ( i t - 2 \alpha \beta N \lambda_j \Big )^{-\beta/2} \, dt,
\end{equation} 
where $c > 2 \alpha \beta N \lambda_{1}$.
 \end{proposition}
 
 \begin{proof}
 This requires only minor modification of the proof of Proposition \ref{P2.1}.
 Specifically, from \cite{DE02} we know that the $\beta$ generalisation of
 (\ref{M1}) is to include a factor of $\prod_{l=1}^N q_l^{\beta - 1}$.
 Including this factor in the integrand on the LHS of (\ref{M2a}) and evaluating implies
the integral in (\ref{A3+}).
 \end{proof}

The considerations of \S \ref{S2.3} also admit $\beta$ generalisations. First, the matrix integral
evaluation (\ref{M4}) extends to a matrix integral over the Haar measure for unitary matrices
 and unitary symplectic matrices,
these diagonalising GUE and GSE matrices respectively. Of particular interest for the future working
of Section \ref{S3.2} is the unitary matrix integral
\begin{equation}\label{A4+}   
  \int_{U } e^{4 \alpha N {\rm Tr} (U \Lambda U^\dagger  \hat{\mathbf 1}_N   \hat{\mathbf 1}_N^T)}
  \, (U^\dagger d U)  \propto
\int_{-\infty-ic}^{\infty -ic} e^{it} \prod_{j=1}^N \Big ( i t - 4 \alpha N \lambda_j \Big )^{-1} \, dt.
 \end{equation}
 
 Generalising (\ref{R1}) to the shifted GUE and GSE ensembles gives a characterisation of the
 corresponding parameter dependent probability density function $p_t$ as satisfying the
 Fokker-Planck equation (\ref{R3}) with the Dyson index $\beta$ appearing as a factor of
 the second appearance of the partial derivative with respect to $\lambda_j$ through a
 multiplication by $1/\beta$; see \cite[\S 11.1]{Fo10}. In the theory of the Fokker-Planck equation,
 $\beta$ then has the interpretation as the inverse temperature, which is a prevalent point
 in the writings of Dyson on random matrices, beginning with \cite{Dy62}. A conjugation of
 the Fokker-Planck operator gives a Schr\"odinger operator of Calogero-Sutherland type
 in imaginary time ---
 see \cite[\S 11.3]{Fo10} --- with the case $\beta = 2$ then corresponding to free fermions.
 The latter also admits an interpretation as non-intersecting Brownian walkers 
 \cite{Ka16}; for a study of outliers in this context see \cite{ADM09}.
 
 In the theory of the scaled distribution of the largest eigenvalue for the Gaussian ensembles,
 as implied by the joint PDF (\ref{A3+}) in the case $\alpha = 0$, well known results due to Tracy
 and Widom (see the review \cite{FW15}) give evaluations in terms of a particular Painlev\'e II transcendent.
 This transcendent is the Hasting-Macleod solution of the Painlev\'e II  equation, specified as satisfying
 \begin{equation}\label{A5+}  
 q'' = s q + 2 q^3, \qquad q(s) \mathop{\sim}\limits_{s \to \infty} {\rm Ai}(s),
 \end{equation}
 where ${\rm Ai}(s)$ denotes the Airy function. Let $p_\beta^{\rm soft}(s)$ denote the PDF of the scaled
 largest eigenvalue, with the scaling defined by the first formula in (\ref{C1}) for $\beta = 1,2$ and by that formula
 with $N \mapsto N/2$ for $\beta = 4$. Denote the corresponding cumulative distribution by $E_\beta^{\rm soft}(s)$
 so that $p_\beta^{\rm soft}(s) = - {d \over d s} E_\beta^{\rm soft}(s)$. The results of Tracy
 and Widom give
  \begin{equation}\label{A6}
  E_2^{\rm soft}(s) = \exp\Big ( - \int_s^\infty (x-s) q^2(x) \, dx \Big )
  \end{equation}
  and
  \begin{equation}\label{A6a}  
   E_1^{\rm soft}(s) = ( E_2^{\rm soft}(s) )^{1/2} \exp \Big ( {1 \over 2}  \int_s^\infty q (x) \, dx \Big ), \quad
     E_4^{\rm soft}(s) = ( E_2^{\rm soft}(s) )^{1/2} \cosh \Big ( {1 \over 2}   \int_s^\infty q (x) \, dx \Big ).
  \end{equation}   
  For the general critical regime scaling, define the scaled variable $w$ as in   (\ref{C1}) for $\beta = 1,2$ and by that formula
 with $N \mapsto N/2$ for $\beta = 4$, in keeping with the prescription relating to the scaling of the largest eigenvalue.
 In the cases $\beta = 2,4$ it turns out that formulas in terms of the transcendent $q(s)$ again hold
 true \cite{BR01a,Ba06,Wa09,BV12,Fo13}, although $q(s)$ must be supplemented by two functions $f=f(s,w)$,
 $g = g(s,w)$ satisfying
   \begin{equation}\label{A7}
{\partial \over \partial w}
\begin{pmatrix} f \\ g \end{pmatrix}  =
\begin{pmatrix} q^2 & - w q - q'  \\
- w q + q'  &  w^2 - s - q^2 \end{pmatrix} \begin{pmatrix} f \\ g \end{pmatrix} ,
 \end{equation}  
subject to the initial conditions
\begin{equation}\label{fgE}
f(s,0) = g(s,0) = E(s), \qquad E(s) := \exp \Big ( - \int_s^\infty q(t) \, dt \Big ).
\end{equation}
The equation (\ref{A7}) is known in the theory of Painlev\'e II as one member of the Lax pair for $q(s)$,
first considered in \cite{FN80}.

\begin{proposition}\label{P2.11}
Specify $F_{\beta, w}(s)$ as the critical regime scaling generalisation of $E_\beta^{\rm soft}(s)$.
We have
\begin{align}
F_{2, w}(s) & = f(s;w) E_2^{\rm soft}(s), \label{fgF} \\
F_{4, w}(s) & = {1 \over 2} \Big ( (f(s;w)+g(s;w)) (E(s))^{-1/2} + (f(s;w)-g(s;w)) (E(s))^{1/2} \Big ) ( E_2^{\rm soft}(s))^{1/2}. \label{fgF1} 
\end{align}
\end{proposition}

\begin{proof}
  Bloemendal and Vir\'ag \cite{BV12} have noted that the validity of these formulas can be
  established by directly checking the characterisation  (\ref{1.K}) of $F_{\beta, w}(x)$.
  \end{proof}

  \begin{remark} 1.~Substituting $w=0$ in (\ref{fgF1}) and comparing with the first formula
  in (\ref{A6a}) shows $F_{4, w}(s) |_{w=0} =  E_1^{\rm soft}(s)$, which in fact can be
  anticipated \cite{Wa09,Fo13}. \\
  2.~Rumanov \cite{Ru15,Ru16} has initiated a program of study on Lax pairs associated with
   (\ref{1.K}) for general even $\beta$, with concrete results obtained for $\beta = 6$. The
   latter have been  further refined in \cite{GIKM16}. \\
   3.~For $s \to - \infty$ it is known \cite{DIK08,BBD07} (see also the review \cite[Eq.~(3.33)]{Fo14})
    \begin{equation}
    E_2^{\rm soft}(s) \mathop{\sim}\limits_{s \to - \infty} e^{- |s|^3/12 - (1/8) \log | s|} 
  \end{equation}
  and \cite{BR01+}  
    \begin{equation}
    f(s,w)  \mathop{\sim}\limits_{s \to - \infty}  e^{-|s|^{3/2}/6 + |x| w/2 - w^2 |x|^{1/2}},
   \end{equation}
   where in both formulas the exponents have been truncated at the constant term (i.e.~term
   independent of $s$).
   Substituting in (\ref{fgF}) gives the left tail asymptotics  of $F_{2, w}(s)$. Note that the resulting formula
   is consistent with (\ref{1.K}).
   \end{remark}
   
   \subsection{Eigenvector overlap}
   For a GOE matrix the eigenvectors $ \hat{\mathbf v}$ say
   are distributed uniformly on the unit sphere in $\mathbb R^N$. One consequence is
   that if we take a particular direction, say $\hat{\mathbf 1}$, and form
    \begin{equation}\label{v1}
   (\hat{\mathbf 1}_N \cdot \hat{\mathbf v})^2,
     \end{equation}
    then averaging over the eigenvectors gives zero. Moreover, for large $N$ this is the value of
  (\ref{v1}) almost surely. Of interest is the value of (\ref{v1}) in relation to the 
  eigenvector corresponding to the largest eigenvalue of the random matrix (\ref{1.6}).
  This was first determined by Benaych-Georges and Nadakuditi \cite{BR11}.
  
  \begin{proposition}\label{P2.13}
  Denote the unit eigenvector corresponding to the largest eigenvalue of (\ref{1.6})
  by $\hat{\mathbf v}$. For $N \to \infty$ we have almost surely
  \begin{equation}\label{v2}
  (\hat{\mathbf 1}_N \cdot \hat{\mathbf v})^2 \to \left \{
  \begin{array}{ll} 0, & 0 \le \alpha \le 1   \\
  1 - 1/\alpha^2, & \alpha >  1. \end{array} \right.
 \end{equation}  
 \end{proposition}
 
 \begin{proof} (Outline)
 Let $\mu$ denote the largest eigenvalue of (\ref{1.6}) with corresponding
unit  eigenvector $\hat{\mathbf v}$.  Rearranging the eigen-equation shows
$$
(\mu \mathbb I_N - \tilde{X}) \hat{\mathbf v} = \alpha \hat{\mathbf 1}_N \hat{\mathbf 1}_N^T \hat{\mathbf v} =
\alpha (  \hat{\mathbf 1}_N^T  \hat{\mathbf v} ) \hat{\mathbf 1}_N,
$$
where the second equality follows from the fact that $ \hat{\mathbf 1}_N^T \hat{\mathbf v} $ is a scalar. This implies
$  \hat{\mathbf v} $ is proportional to $  (\mu \mathbb I_N - \tilde{X})^{-1}  \hat{\mathbf 1}_N$. Moreover, the
proportionality can be specified by the fact that $\hat{\mathbf v} $ is a unit vector. Thus
$$
 \hat{\mathbf v} ={1 \over c}  (\mu \mathbb I_N - \tilde{X})^{-1}  \hat{\mathbf 1}_N, \qquad c = ( \hat{\mathbf 1}_N^T ( \mu \mathbb I_N - \tilde{X})^{-2}  \hat{\mathbf 1}_N )^{1/2}.
 $$
 
 Next diagonalise the GOE matrix $\tilde{X}$ using $\tilde{X} = U \Lambda U^T$ as in the proof of Proposition
 \ref{P2.1}.  With $\hat{\mathbf w} = U^T  \hat{\mathbf 1}_N$ this shows
 $$
 \hat{\mathbf v} = U ( \mu \mathbb I_N - \tilde{X})^{-1} \hat{\mathbf w}, 
 \qquad c = ( \hat{\mathbf w}^T ( \mu \mathbb I_N - \Lambda)^{-2}  \hat{\mathbf w} )^{1/2}.
 $$
 It then follows 
 $$
   (\hat{\mathbf 1} \cdot \hat{\mathbf v})^2 = {1 \over c^2 } ( \hat{\mathbf w}^T ( \mu \mathbb I_N - \Lambda)^{-1}  \hat{\mathbf w} )^{2}.
   $$
   To close out the proof from here we need the fact that for $N$ large, and $\rho^{\rm W}(x)$ denoting the Wigner semi-circle (\ref{1.6a}),
   almost surely
   $$
   \hat{\mathbf w}^T ( \mu \mathbb I_N - \Lambda)^{-1}  \hat{\mathbf w}  \to \int_{-1}^1 {\rho^{\rm W}(x) \over \mu - x} \, dx.
   $$
 Note that in the  proof of Proposition
 \ref{P2.1} is established upon averaging. The proportionality $c$ is just minus the derivative with respect to $\mu$ of this for $\mu > 1$
 Recalling now (\ref{1.8g}) gives (\ref{v2}) for $\mu \ge 1$, with the cases $0 \le \alpha \le 1$ obtained  by taking $\lim_{\mu \to 1^+}$.
 \end{proof}
 
 \begin{remark}
Recently Bao and Wang \cite{BW21} have studied the eigenvector overlap for the random matrix (\ref{1.6}), with
$\tilde{G}$ therein a scaled GUE matrix, in the critical regime as specified by the scaling (\ref{C1}). Specifically, they studied
the first component $x_j^{(1)}$ of the eigenvector $\mathbf x_j$ corresponding to the $j$-th largest eigenvalue.
With $\{ \sigma_j \}_{j=1}^N$ the eigenvalues of (\ref{1.6}) and $\{\mu_j\}_{j=1}^{N-1}$ the eigenvalues of this same random
matrix with the first row and column removed, the starting point of their analysis is the identity
  \begin{equation}\label{rn}
| x_j^{(1)}|^2 = \prod_{i=1}^{j=1} {\sigma_j - \mu_j \over \sigma_j - \sigma_i} \prod_{i=j+1}^N {\sigma_j - \mu_{i-1} \over \sigma_j - \sigma_i}.
 \end{equation}
 Their analysis reveals that for $j$ fixed and in the critical regime, $N^{1/3}   | x_j^{(1)}|^2$ has a well defined limit proportional in
 distribution to the RHS of (\ref{rn}) with the eigenvalues therein replaced by their scaled critical regime counterparts. \\
\end{remark}

\section{A multiplicative rank $1$ perturbation for the LUE}\label{S3}
\subsection{Reduction to an additive rank $1$ perturbation}
Let $X$ be a standard complex Gaussian matrix of size $n \times N$,
$(n \ge N)$, and form the matrix $W = X^\dagger X$. The matrices
$\{ W \}$ are said to be particular complex Wishart matrices (specifically
such that $X$ has mean zero, and the covariance matrix associated with
$W$ is the identity) and are also matrix realisations of the Laguerre
unitary ensemble (LUE) in the case of the Laguerre parameter $a=n-N$.
For general Laguerre parameter $a > -1$, the LUE can be specified by the
eigenvalue PDF proportional to
\begin{equation}\label{3.1}
\prod_{l=1}^N x_l^a  e^{-x_l} \prod_{1 \le j < k \le N} (x_k - x_j)^2, \quad x_l \in \mathbb R^+.
\end{equation}
It follows that for $a=n-N$ a viewpoint on (\ref{3.1}) is as the PDF for the squared singular values
of $X$.

For $\Sigma$ and $N \times N$ positive definite matrix, construct from $X$ a correlated
complex Gaussian matrix 
\begin{equation}\label{3.1a}
\tilde{X} = X \Sigma^{1/2},
\end{equation}
and use it in turn to construct a correlation complex Wishart matrix
\begin{equation}\label{3.1b}
\tilde{W} =  \Sigma^{1/2} X^\dagger X \Sigma^{1/2}.
\end{equation}
We see that the choice
\begin{equation}\label{3.1b+}
\Sigma = {\rm diag} \, (b,1,\dots,1), \quad b > 0,
\end{equation}
corresponds to a rank $1$ multiplicative perturbation of $X$. In fact from the
viewpoint of eigenvalues, the corresponding multiplicative perturbation of $\tilde{W}$ can
be written as a rank $1$  additive perturbation. To see this we note that the
nonzero eigenvalues of $ \Sigma^{1/2} X^\dagger X \Sigma^{1/2}$ are the same
as those for $ X \Sigma X^\dagger$ (more generally $AB$ and $BA$ have the same
nonzero eigenvalues). But with $\Sigma$ given by (\ref{3.1b+}),
\begin{equation}\label{3.1c}
X \Sigma X^\dagger = X_1 X_1^\dagger + b \mathbf x \mathbf x^\dagger,
\end{equation}
where $X_1$ refers to $X$ with the first column deleted, and $\mathbf x$
denotes the first column of $X$. A strategy analogous to that used in the proof
of Proposition \ref{P2.1} now suffices to specify an eigenvalue separation
effect as a function of $b$ \cite{BBP05,BS06,BFF09}.

 \begin{proposition}\label{p3.1}
 Consider the particular correlated complex Wishart matrix specified by (\ref{3.1b}) and 
  (\ref{3.1b+}), and scale the eigenvalues by dividing by $n$. In the limit $n \to \infty$
  with $n/N := \gamma \ge 1$ fixed. Provided $b > 1 + \sqrt{\gamma}$, a single eigenvalue
  separates from the upper endpoint of the support $(1 + 1/\sqrt{\gamma})^2$, and occurs
  at the point
 \begin{equation}\label{3.1d} 
b \Big ( 1 + {\gamma^{-1} \over b - 1} \Big ).
\end{equation}
\end{proposition}

\begin{proof} (Sketch)
When divided by $n$, replacing the Wigner semi-circle (\ref{1.6a}) for the
normalised density of $ X_1 X_1^\dagger $ in (\ref{3.1c}) is the Mar\v{c}enko--Pastur
functional form \cite{PS11}
 \begin{equation}\label{MP}
 \Big (1 - {1 \over \gamma} \Big ) \delta(x) + {\sqrt{(x - c)(d - x)} \over 2 \pi \gamma x} \delta_{c < x < d},
 \end{equation}
where $c = (1 - \sqrt{\gamma})^2$ and $d = (1 + \sqrt{\gamma})^2$. Here the delta function is in
keeping with the fraction of zero eigenvalues of $X_1 X_1^\dagger $ equalling $(1-1/\gamma)$.
The argument of the working of the proof of Proposition \ref{P2.1} gives that the secular equation
for the eigenvalues of the RHS of (\ref{3.1c}), divided by $n$, reads
 \begin{equation}\label{MP1}
1 = {b (1 - 1/\gamma) \over \lambda} + b \int_c^d {\tilde{\rho}^{\rm MP}_{(1)}(x) \over \lambda - x} \, dx,
 \end{equation}
where $\tilde{\rho}^{\rm MP}_{(1)}(x)$ denotes the second term in (\ref{MP}); cf.~(\ref{1.8g}). For the
integral in (\ref{MP1}) we have the evaluation (see e.g.~\cite[Eq.~(2.24)]{BFF09})
$$
{1 \over 2 \gamma} \bigg ( 1 - {\gamma - 1 \over z}- \Big ( 1 - {2 (\gamma + 1) \over z} +
{(\gamma - 1)^2 \over z^2} \Big )^{1/2} \bigg ).
$$
Substituting in (\ref{MP1}), and observing that both terms therein are decreasing
functions of $\lambda$ and so take their maximum value when $\lambda=d$ gives the stated condition
for eigenvalue separation, while solving for $\lambda$ under this condition gives the value
(\ref{3.1d}).
\end{proof}

An illustration of the prediction of Proposition \ref{p3.1} is given in Figure \ref{F3.1}. A comprehensive study of
this phase transition effect, extended to (\ref{3.1b}) with the parameter $b$ repeated $r$ times down the
diagonal of $\Sigma$ and including the critical regime (see subsection \ref{S3.3} below) was undertaken
by Baik, Ben Arous and P\'eche \cite{BBP05}. Subsequently, it has been customary to use
the term BBP transition in this context.

 \begin{figure*}
\centering
\includegraphics[width=0.95\textwidth]{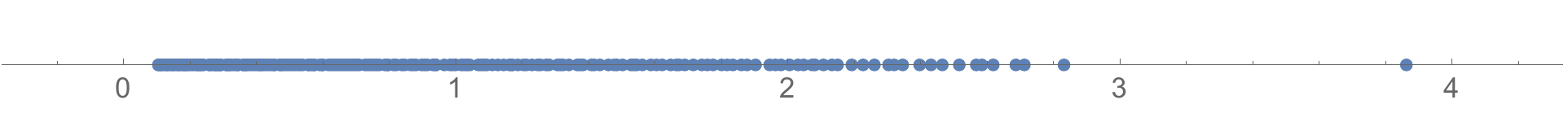}
\caption{Eigenvalues divided by 400 of a $200 \times 400$ sample random matrix as specified by (\ref{3.1b}) and 
  (\ref{3.1b+}) with $b=3$. Theoretically the bulk of the spectrum has support approximately $[0.09, 2.9]$, with the
  outlier at $15/4$. }
  \label{F3.1}
\end{figure*}

  \begin{remark}
  1.~(Rank 1 update)
 Generalise (\ref{3.1b+}) so that $\Sigma = {\rm diag} \, (b_1,b_2,\dots,b_N)$ with each $b_i > 0$.
 Specify $X$ as in (\ref{3.1c}), denote by $X_j$ the matrix obtained from $X$ by deleting the first $j$ columns
 $\mathbf x_1,\dots, \mathbf x_j$, and denote by $\Sigma_j$ the $(N - j) \times (N - j)$ matrix
 ${\rm diag} \, (b_{j+1},\dots, b_N)$. In keeping with (\ref{3.1c}) we have
 $$
 X_{j-1} \Sigma_{j-1} X_{j-1}^\dagger = X_j \Sigma_j X_j^\dagger + b_j \mathbf x_j \mathbf x_j^\dagger \qquad (j=1,\dots,N).
 $$
 Iterating this backwards, $j=N,N-1,\dots,1$ gives a rank $1$ update constuction of $X \Sigma X^\dagger$; recall the paragraph
 including (\ref{1.7}). \\
  2.~(Note on eigenvector overlap)
  Before the derivation of Proposition \ref{P2.13} the eigenvalue overlap associated with the rank $1$ perturbation
  (\ref{3.1c}) was calculated by Paul \cite{Pa07}. This was in the case that $X$ therein has real rather than complex entries, but this has no effect on 
  the result. The method of the proof of Proposition \ref{P2.13} carries over, with the role of (\ref{1.8g}) now played by (\ref{MP1}).
  To state the result,
  denote the unit eigenvector corresponding to the largest eigenvalue of (\ref{3.1c})
  by $\hat{\mathbf v}$. Then for $N \to \infty$ we have the almost surely convergence
  \begin{equation}\label{v2}
  |\hat{\mathbf 1}_N \cdot \hat{\mathbf v}|^2 \to \left \{
  \begin{array}{ll} \displaystyle {(b-1)^2 - \gamma \over (b-1)^2 + \gamma (b-1)} , & b > 1 + \sqrt{\gamma}    \\
  0, &  {\rm otherwise}. \end{array} \right.
 \end{equation}    
 The recent work \cite{DDC21} considers this overlap in the case of complex entries for the eigenvector corresponding
 to the smallest eigenvalue, and shows that when multiplied by $N$ it has the limiting distribution
 $\chi_2^2/(2b)$, with $\chi_2^2$ the chi-square random variable with two degrees of freedom.
\end{remark}  

\subsection{An application of the HCIZ matrix integral}\label{S3.2}
It follows from (\ref{3.1b}) and the definition of $X$ therein that the correlated complex Wishart matrix is
specified by a PDF proportional to
 \begin{equation}\label{4.1a} 
 \Big ( {1 \over \det  \Sigma} \Big )^{1/2} e^{- {\rm Tr} \, (\tilde{X}^\dagger \tilde{X} \Sigma^{-1})}.
  \end{equation}
  A fundamental change of variables in random matrix theory, see e.g.~\cite[Eq.~(3.23)]{Fo10}, gives that
  $\tilde{W} = \tilde{X}^\dagger \tilde{X}$ then has PDF proportional to
  $$
   \Big ( {1 \over \det  \Sigma} \Big )^{1/2}  \Big ( \det \tilde{W} \Big )^{(n-N)} e^{- {\rm Tr} \, (\tilde{W} \Sigma^{-1})}.
   $$
   Diagonalising the complex Hermitian matrix $\tilde{W}$ according to $\tilde{W} = U \Lambda U^\dagger$ for
   $U \in U(N)$ and making use of the corresponding change of variables formula (see \cite[Prop.~1.3.4]{Fo10})
   gives that the eigenvalue PDF of $\tilde{W}$ is proportional to
 \begin{equation}\label{4.1b}    
 \Big ( {1 \over \det  \Sigma} \Big )^{1/2}  \prod_{l=1}^N \lambda_l^{n-N} e^{- \lambda_l} \prod_{1 \le j < k \le N}
 (\lambda_k - \lambda_j)^2 \int_{U \in {\rm U}(N)} e^{- {\rm Tr} \, (U \Lambda U^\dagger (\Sigma^{-1} - \mathbb I_N))}
 \, (U^\dagger d U).
   \end{equation}
   
   For $A,B$ Hermitian matrices of size $N \times N$, the HCIZ matrix integral (named after Harish-Chandra
   \cite{HC57}, and Itzykson and Zuber \cite{IZ80})
   \begin{equation}\label{HCIZ}
\int_U \exp (U^\dagger A U B) \, [U^\dagger dU] = {\prod_{j=1}^N} \Gamma(j) \,
{\det [ e^{ a_j b_k} ]_{j,k=1}^N \over \Delta_N(a)  \Delta_N(b)},
\end{equation}
where $[U^\dagger dU]$ denotes the normalised Haar measure for $U(n)$, and for an array $x = (x_1,\dots, x_N)$,
$\Delta_N(x) := \prod_{1 \le j < k \le N} (x_k - x_j)$.
Application of (\ref{HCIZ}) allows the matrix integral in (\ref{4.1b}) to be calculated.

 \begin{proposition}\label{p3.2}
 In the case that $\Sigma$ is given by (\ref{3.1b+}) the PDF (\ref{4.1b}) of $\tilde{W}$ simplifies to be
 proportional to 
  \begin{equation}\label{4.1c} 
 \prod_{l=1}^N \lambda_l^{n-N} e^{- \lambda_l} \prod_{1 \le j < k \le N}
 (\lambda_k - \lambda_j) \det \Big [ [\lambda_j^{k-1}]_{j=1,\dots, N \atop k=1,\dots,N-1} \: [e^{c \lambda_j}]_{j=1}^N \Big ],
 \end{equation}
 where $c:=1-1/b$.
 \end{proposition}
 
 \begin{proof}
 Write the eigenvalues of $(\Sigma^{-1} - \mathbb I_N)$ as $c_1,c_2,\dots,c_N$. We then see that (\ref{4.1c}) follows from
 (\ref{4.1b}) by making use of (\ref{HCIZ}) and the limit formula
 $$
 \lim_{c_1,\dots,c_{N-1} \to 0} 
 {\det [ e^{ \lambda_j c_k} ]_{j,k=1}^N \over   \Delta_n(c)} \propto {1 \over c_N^{N-1}}
 \det \Big [ [\lambda_j^{k-1}]_{j=1,\dots, N \atop k=1,\dots,N-1} \: [e^{c_N \lambda_j}]_{j=1}^N \Big ].
 $$
 This limit formula in turn follows by taking the limits in order $c_k \to 0$, $(k=1,\dots,N-1)$ applied to column
 $k$ of the determinant, after first subtracting multiplies of the limiting value of the earlier columns so that the
 first $(k-1)$ terms of the power series expansion have been eliminated.
 \end{proof}
 
 \begin{remark}\label{R3.3}
 The above working implies that for $B$ of rank $1$, with its nonzero eigenvalue equal to $b$,
   \begin{equation}\label{HCIZ1}
\int_U \exp (U^\dagger A U B) \, [U^\dagger dU] \propto
{\det [ [a_j^{k-1}]_{j=1,\dots,N \atop k=1,\dots,N-1}   [e^{ a_j b}]_{j=1}^N ] \over \Delta_n(a)  b^{N-1}}.
\end{equation}
On the other hand (\ref{A4+}) tells us that this same matrix integral is proportional to
\begin{equation}\label{A4+d}   
\int_{-\infty-ic}^{\infty -ic} e^{it} \prod_{j=1}^N \Big ( i t - b a_j \Big )^{-1} \, dt.
 \end{equation}
 Indeed one can check that computing this contour integral using residues gives the same functional
 form as expanding the determinant in (\ref{HCIZ1}) by the final column and simplifying using
 the Vandermonde determinant formula. It follows that in the case $\beta = 2$ the eigenvalue PDF
 (\ref{A3+}) for the perturbed GUE can also be written, up to proportionality, as
 \begin{equation}\label{A3+1}   
\prod_{j=1}^N e^{- \beta N \lambda_j^2} \prod_{1 \le j < k \le N} ( \lambda_j - \lambda_k)
 \det \Big [ [\lambda_j^{k-1}]_{j=1,\dots, N \atop k=1,\dots,N-1} \: [e^{4 \alpha N  \lambda_j}]_{j=1}^N \Big ].
\end{equation} 
 \end{remark}
 
 \subsection{A joint eigenvalue PDF}\label{S3.2a}
 Starting from (\ref{3.1c}), it is possible to compute the joint eigenvalue PDF for the eigenvalues
 of both $X_1 X_1^\dagger$ and $X \Sigma X^\dagger$, following \cite[\S 3.1]{Fo13}.
 Denoting the $N-1$ nonzero eigenvalues of the $n \times n$ matrix
 $X_1 X_1^\dagger$ by $\{ \mu_j \}_{j=1}^{N-1}$ and ordered as in
 (\ref{1.8d}), from the derivation of (\ref{1.8c}) we have that the equation determining the
 eigenvalues of $X \Sigma X^\dagger$ is
 \begin{equation}\label{1.8c+}
0 = 1   - b \bigg (   {u_0 \over \lambda} 
+   \sum_{j=1}^{N - 1}{u_j \over \lambda -  \mu_j } \bigg ), \qquad u_0 = \sum_{j=N}^{n-N-1}  |x^{(j)}|^2 ,
\: \: u_j =  |x ^{(j)}|^2. 
 \end{equation}
Since $\mathbf x$ is a standard complex vector, we have that $u_0$ is distributed as $\Gamma[n-N-1,1]$,
and each $u_j$ as $\Gamma[1,1]$.
 The joint distribution of $ \{u_j\}_{j=0}^{N-1}$ is
 therefore proportional to 
  \begin{equation}\label{1.8c1}
 u_0^{n-N-1} e^{- u_0} \prod_{l=1}^{N-1} e^{- u_l}.
  \end{equation} 
  Regarding these variables as the residues in the random rational function specified by  the RHS of
  (\ref{1.8c+}), we have that $ \{ \mu_j \}_{j=1}^{N-1}$ are the poles, while the eigenvalues of
 $X \Sigma X^\dagger$ are the zeros. Let the latter be denoted $\{\lambda_j\}_{j=1}^N$, which we know
 must be interlaced as in (\ref{1.8e}). In terms of the zeros and the poles we have 
   \begin{equation}\label{1.8c2}
 1 - b \bigg (   {u_0 \over \lambda} 
+   \sum_{j=1}^{N - 1}{u_j \over \lambda -  \mu_j } \bigg )
 = {\prod_{l=1}^N (\lambda - \lambda_l) \over   \lambda \prod_{l=1}^{N-1} (\lambda - \mu_l)}.
 \end{equation}  
 For given $\mu_j$, computing the Jacobian for the change of variables from
 residues to the zeros gives a particular conditional PDF.
 
 \begin{proposition}
 Let the PDF for $\{u_j\}$ to given by (\ref{1.8c1}). The PDF for  $\{\lambda_j\}$ with $\{ \mu_j \}$ given is proportional to
  \begin{equation}\label{1.8c3}
 \prod_{l=1}^{N-1} \mu_l^{-n + N - 1} e^{\mu_l/b}  \prod_{k=1}^N \lambda_k^{n - N} e^{- \lambda_k/b} 
{ \prod_{1 \le j < k \le N} ( \lambda_j - \lambda_k) \over  \prod_{1 \le j < k \le N - 1} ( \mu_j - \mu_k)}
   \end{equation} 
   supported on (\ref{1.8e}) with $\mu_N = 0$.
 \end{proposition}     
 
 \begin{proof}
 Expanding both sides of (\ref{1.8c2}) in powers of $1/\lambda$ and equating the coefficient of $1/\lambda$ shows
 $$
  b  \sum_{j=0}^{N-1} u_j =  \sum_{j=1}^N  \lambda_j -    \sum_{j=1}^{N-1} \mu_j,
 $$
 which when substituted in (\ref{1.8c1}) accounts for the exponential term in (\ref{1.8c3}). It remains to compute
 the Jacobian. For this purpose, note from  (\ref{1.8c2}) by
 computing residues that
 $$
b u_0   = { \prod_{l=1}^N \lambda_l \over \prod_{l=1}^{N-1} \mu_l}, \qquad
- b u_j = {\prod_{l=1}^N (\mu_j - \lambda_l) \over   \prod_{l=1, l \ne j }^N (\mu_j - \mu_l)}.
 $$
 This shows that up to a possible sign, and with $\lambda_N = 0$,
   \begin{equation}\label{1.8c3+}
 \det \bigg [ {\partial u_{j-1} \over \partial \lambda_k} \bigg ]_{j,k=1}^N = {1 \over b^N} {\prod_{j,l=1}^N (\mu_j - \lambda_l) \over \prod_{j,l=1 \atop l \ne j}^N (\mu_j - \mu_l) }
 \det \Big [ {1 \over \mu_j - \lambda_k} \Big ]_{j,k=1}^N.
  \end{equation} 
 The determinant on the RHS is known of the Cauchy double alternant and has an evaluation in terms of products
 (see e.g.~\cite[Eq.~(4.33)]{Fo10}) which implies the remaining terms in (\ref{1.8c3}).
 \end{proof}

 In the setting of (\ref{3.1c}) the given eigenvalues $\{\mu_j\}$ in (\ref{1.8c3}) have the PDF (\ref{3.1}) with 
 $N$ replaced by $N-1$ and $a = n - N + 1$.
 Hence the joint PDF of the eigenvalues of the matrices $X_1 X_1^\dagger$ and $X \Sigma X^\dagger$ in (\ref{3.1c})
 is proportional to
 \begin{equation}\label{1.8c3-}
 \prod_{l=1}^{N-1}e^{-\mu_l(1 - 1/b)}  \prod_{k=1}^N \lambda_k^{n - N} e^{-\lambda_k/b} 
\prod_{1 \le j < k \le N} ( \lambda_j - \lambda_k)   \prod_{1 \le j < k \le N - 1} ( \mu_j - \mu_k)
   \end{equation} 
   again with the requirement of the interlacing (\ref{1.8e}) with $\lambda_N = 0$.
 The latter
  the ordering of $\{ \mu_j \}_{j=1}^N$  (\ref{1.8d}), the function of  $\{ \lambda_j \}_{j=1}^N$ defined by the interlacing when viewed
 as an indicator function has the determinantal form
  \begin{equation}\label{1.8c4}
  \det [ \chi_{\lambda_j - \mu_k > 0} ]_{j,k=1}^N  \qquad     \chi_A = \bigg \{ \begin{array}{cc}  1 & A \, {\rm true} \\
  0 & A \, {\rm false}. \end{array} 
  \end{equation}  
 Including this as a factor in (\ref{1.8c3-}) allows the PDF for $\{\mu_j\}$ to be computed by integrating
 each $\mu_j$ over $\mathbb R^+$. These integrations can be done can be done with the aid
 of a minor variant of Andr\'eief's identity (see \cite{Fo18}) which shows
 \begin{multline}\label{1.8c5} 
 {1 \over (N-1)!} \int_0^\infty d \mu_1 \cdots  \int_0^\infty d \mu_{N-1} \, \prod_{l=1}^{N -1} e^{-(1-1/b)\mu_l}
    \prod_{1 \le j < k \le N-1}  ( \mu_k - \mu_j) 
     \det [ 
    \chi_{\lambda_j - \mu_k > 0} ]_{j,k=1}^N \\ = \det \bigg [ \Big [ \int_0^{\lambda_j}   \mu^{k-1}  e^{-\mu(1-1/b)}  \, d \mu \Big  ]_{j=1,\dots,N \atop k = 1,\dots, N - 1} \: [1]_{j=1,\dots,N}  \bigg ] \\ \propto \prod_{j=1}^N e^{- \mu (1 - 1/b) \lambda_j}
    \det \bigg [ \Big [  [\lambda_j^{k-1}]_{j=1,\dots,N \atop k=1,\dots,N-1} \: [ e^{\mu (1 - 1/b) \lambda_j} ]_{j=1,\dots,N} \bigg ],
     \end{multline}
     where the final expression follows using integration by parts and elementary column operations.
     Replacing the terms dependent on $\{ \mu_j \}$ in (\ref{1.8c3-}) by the final expression in (\ref{1.8c5}) reclaims
     (\ref{4.1c}) for the marginal PDF of  $\{\lambda_j\}$.
     
     \begin{remark}
     Consider the setting of (\ref{3.1c}) with $X$ of size $N \times (N + 1)$, $\Sigma$ of size $(N + 1) \times (N + 1)$ and $X_1$
     of size $N \times N$. Denote the eigenvalues of $X \Sigma X^\dagger$  $(X_1 X_1^\dagger)$ by $\{\lambda_j \}$ and
  ($\{ \mu_j  \})$. Repeating the considerations which lead to (\ref{1.8c3-}) shows that the joint eigenvalue PDF is proportional to
   \begin{equation}\label{1.8c3f}
 \prod_{l=1}^{N-1}e^{-\mu_l(1 - 1/b)}  \prod_{k=1}^Ne^{-\lambda_k/b} 
\prod_{1 \le j < k \le N} ( \lambda_j - \lambda_k)  ( \mu_j - \mu_k)
   \end{equation} 
   subject to the interlacing  (\ref{1.8e}). This PDF first appeared in the study of probabilistic models
   related to the longest increasing subsequence of a random permutation
   \cite{BR01a}; see also \cite{FR02b}.
   \end{remark} 
 
 \subsection{Correlation kernel for the soft edge critical regime}\label{S3.3}
 Let $p_N(x_1,\dots,x_N)$ denote an eigenvalue PDF supported on $I \subset \mathbb R$. The $k$-point correlation
 function $\rho_{(k)}(x_1,\dots,x_k)$ is specified in terms of $p_N$ by
 $$
 \rho_{(k)}(x_1,\dots,x_k) = {N! \over (N-k)!} \int_I dx_{k+1} \cdots \int_I dx_N \, p_N(x_1,\dots,x_N).
 $$
 Note that the case $k=1$ corresponds to the eigenvalue density. The eigenvalue PDFs (\ref{4.1c})
 and (\ref{A3+1}) correspond to a determinantal point processes. This  means that $ \rho_{(k)}$ can be
 expressed in determinant form
   \begin{equation}\label{B0}
 \rho_{(k)}(x_1,\dots,x_k) =  \det [ K_N(x_j,x_l) ]_{j,l=1,\dots,k},
  \end{equation} 
 where $K_N(x,y)$ --- referred to as the correlation kernel --- can be expressed in terms of certain
 orthogonal polynomials and special functions. In keeping with the focus of this section on perturbation
 of  the LUE, we consider (\ref{4.1c}). Relevant for this is the particular non-symmetric Laguerre polynomial
 kernel
  \begin{equation}\label{B1}
  K_n^a(x,y) = y^a e^{-y} \sum_{p=0}^{n-1} {(p+a)! \over p!} L_p^a(x) L_p^a(y).
 \end{equation}  
 It is a standard result in random matrix theory that substituting (\ref{B1}) with $n=N$ in (\ref{B0}) gives
 $ \rho_{(k)}$ for the unperturbed LUE (\ref{3.1}); see \cite[\S 5.1.2]{Fo10}. Also relevant are the so-called incomplete
 multiple Laguerre functions of type I and II \cite{BK04,DF06}
  \begin{equation}\label{B2}
\tilde{\Lambda}^{(1)}(x)  = 
\int_{{\mathcal C}_{\{0,-c\}}}
{e^{-xz} (1 + z)^{N+a} \over z^{N-1} (z + c)} \,
{dz \over 2 \pi i}, \quad
{\Lambda}^{(1)}(x)  = 
\int_{{\mathcal C}_{\{-1\}}} {e^{xz} z^{N-r}\over
(1 + z)^{N+a} } \, {dz \over 2 \pi i}. \\
\end{equation}
Here ${\mathcal C}_{\{0,-c\}}$, ${\mathcal C}_{\{-1\}}$ are simple contours
encircling the points $\{0,-c\}$ and $\{-1\}$ respectively.

\begin{proposition}
The $k$-point correlation function for the PDF (\ref{4.1c}) with $n-N = a$ is given by (\ref{B0}) with
  \begin{equation}\label{B3}
  K_N(x,y) =  K_{N-1}^{a+1}(x,y)  + \tilde{\Lambda}^{(1)}(x)  {\Lambda}^{(1)}(y).
 \end{equation}
 \end{proposition}  
  
 \begin{remark}
 This kernel is a special case of the correlation kernel for the PDF (\ref{4.1c}) with the determinant
 factor therein replaced by $\det [ e^{c_k \lambda_j} ]_{j,k=1,\dots,N}$. A double contour form of
 the determinant in this more general case was first given in \cite{BBP05}, and rederived in the
 context of multiple orthogonal polynomials in \cite{DF06}. Choosing all but $r$ of the $\{c_k\}$
 equal to zero was then shown in the latter reference to allow an evaluation in terms of
the nonsymmetric Laguerre kernel  $K_{N-r}^{a+r}(x,y)$, plus a sum of $r$ terms involving
incomplete
 multiple Laguerre functions of type I and II, which for $r=1$ is (\ref{B3}).
 \end{remark}
 
 Common to both (\ref{4.1c})
 and (\ref{A3+1})  is that there is a tuning of the parameters $c$ and $\alpha$ respectively so that the
 statistical state corresponding to the critical regime --- recall \S \ref{S2.4} --- is identical for both.
 For (\ref{A3+1})  the required scaling is given by (\ref{C1}) while $w,x_1$ for  (\ref{4.1c})
 with $n,N \to \infty$, $n/N = \gamma \ge 1$ fixed are specified by
   \begin{equation}\label{B4}
{N^{1/3} \over (1 + \sqrt{1/\gamma})^{2/3}} \Big ( 1 - \gamma {c \over c + 1} \Big ) \to w, \qquad
{N^{2/3} \sqrt{1/\gamma} \over (1 + \sqrt{1/\gamma})^{4/3}} \Big (
{\lambda_1 \over N} - (1 + \sqrt{\gamma})^2 \Big ) \to x_1.
 \end{equation}
 Applying the scaling (\ref{B4}) to (\ref{B3}), or the scaling (\ref{C1}) to the analogue of (\ref{B3})
 for the PDF (\ref{A3+1}) gives a functional form involving Airy functions \cite{BBP05,DF06}.
 
 \begin{proposition}
 The $k$-point correlation function for the PDF (\ref{4.1c}) with the scaling (\ref{B4}) is
 given by (\ref{B0}) with correlation kernel
  \begin{equation}\label{B5}
  K^{\rm soft,c}(x,y;w) =  K^{\rm soft}(x,y)  + {\rm Ai}(y) \int_{-\infty}^x e^{-w (x - t)} {\rm Ai}(t) \, dt,
 \end{equation}
 where
  \begin{equation}\label{B5a}
K^{\rm soft}(x,y)  = {{\rm Ai}(x)   {\rm Ai}'(y)  - {\rm Ai}(y)   {\rm Ai}'(x)   \over x - y}.
 \end{equation}
  \end{proposition}  

 \begin{remark} 1.~The correlation kernel (\ref{B5a}) is well known in random matrix theory
 as specifying the scaled state in the neighbourhood of the largest eigenvalue for the
 GUE and LUE \cite{Fo93a}, and in fact for a much broader class of random matrices relating to
 Hermitian random matrices with complex entries \cite{PS11}.   Note that
 $ K^{\rm soft,c}(x,y;w) \to K^{\rm soft}(x,y)$ as $ w \to \infty$. \\
 2.~The density is given by setting $x=y$ in $K^{\rm soft,c}(x,y;w) $. Using the integral
 $\int_{-\infty}^\infty e^{wt} {\rm Ai}\,(t) \, dt = e^{w^3/3}$ shows
  \begin{equation}\label{B5b}
  \rho_{(1)}^{\rm soft}(x) \mathop{\sim}\limits_{x \to \infty} {\rm Ai}(x) e^{-wx + w^3/3},
   \end{equation}
   which from general considerations (see \cite[\S 3.5]{Fo14}) coincides with the right tail of
   the PDF corresponding to $F_{2,w}(x)$.
  \end{remark}
  
  The fact that the statistical state of the critical regime is a determinantal point process implies
  a formula for the cumulative distribution function $F^{2,w}(s)$ in terms of a 
  Fredholm determinant,
  \begin{equation}\label{C1a}
  F^{2,w}(s) = \det ( \mathbb I - \mathbb K_2^{2,w});
 \end{equation}    
 see e.g.~\cite[\S 9.1]{Fo10} for the general theory.
Here $ \mathbb K_2^{2,w}$ is the integral operator on $(s,\infty)$ with kernel $K^{\rm soft,c}(x,y;w)$.
As made explicit by Bornemann \cite{Bo09}, there are advantages in using the Fredholm 
determinant for a numerical
tabulation rather than the Painlev\'e expression (\ref{fgF}). 
An exception is the case $w=0$. The, according to (\ref{fgF}),  (\ref{fgE}) and (\ref{A6a})
  \begin{equation}\label{C1b}
   F^{2,w}(s) \Big |_{w=0} = \Big ( E_1^{\rm soft}(s) \Big )^2,
   \end{equation}  
   which can be anticipated already at the finite $N$ level \cite[Eq.~(5.8)]{FR02}.
 The significance of this is that $E_1^{\rm soft}(s)$ and the corresponding PDF,
 which correspond to Tracy--Widom $\beta = 1$, are
 now part of standard software. A tabulation of the PDF corresponding to (\ref{C1b}),
 compared against a simulation based on (\ref{Tb}), is 
given in Figure \ref{F5}.

 \begin{figure*}
\centering
\includegraphics[width=0.65\textwidth]{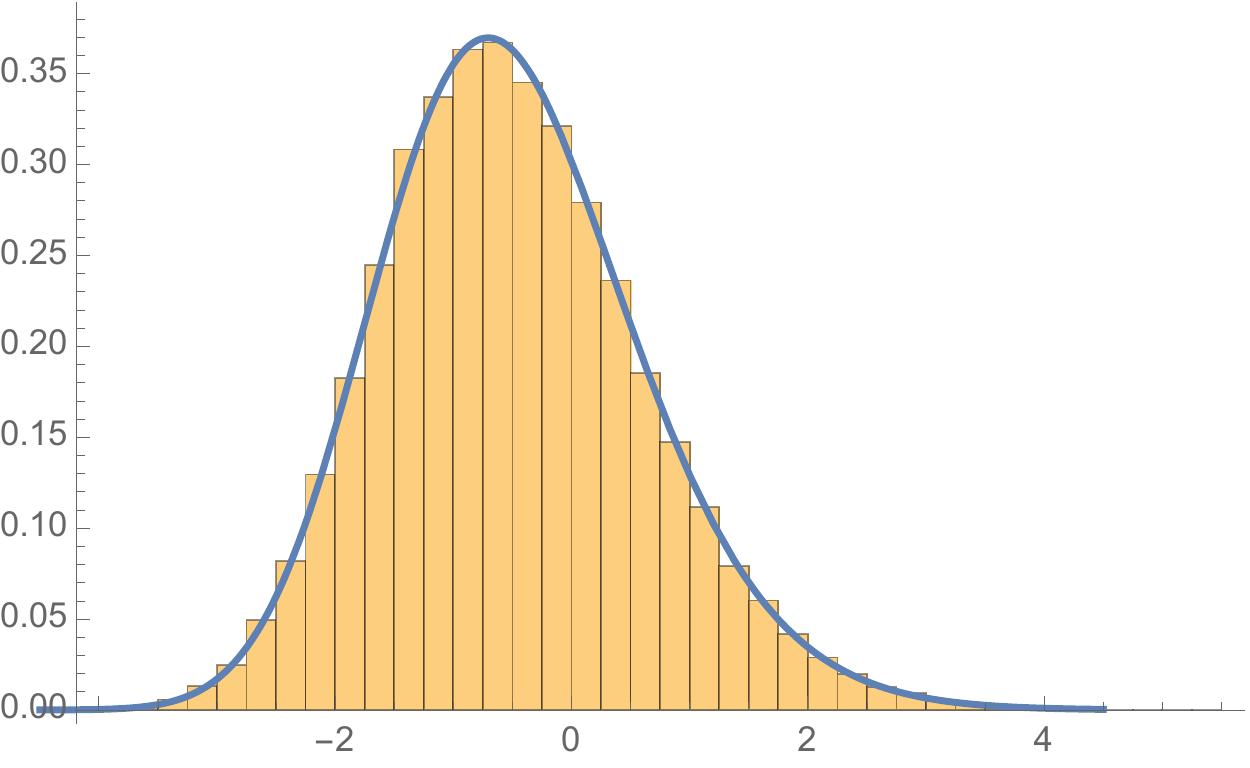}
\caption{Simulated histogram of the PDF for the largest
eigenvalue of (\ref{Tb}) with $\beta = 2$ in the scaled critical regime with
$N = 10^5$,  $\alpha = 1/2$ (equivalently $w=0$), and $M = 5*10^4$ repetitions. The solid
curve is the theoretical graph obtained by  taking minus the derivative of (\ref{C1a}) with respect to $s$.}
  \label{F5}
\end{figure*}

\begin{remark}\label{R3.8}
The statistical state of the $\beta = 4$ soft edge statistical regime is known to form a Pfaffian point process
\cite[Prop.~16]{FR02}, meaning
that the general $k$-point correlation function can be written in the form
\begin{equation}\label{C1c}
\rho_{(k)}(x_1,\dots,x_k) = {\rm Pf} \, (A Z_{2k} ), \qquad A = \begin{bmatrix}
\begin{bmatrix}  f^{11}(x_j, x_l) & f^{12}(x_j, x_l)   \\
 f^{21}(x_j, x_l)  &  f^{11}(x_l, x_j)  \end{bmatrix}_{j,l=1,\dots, k}   \end{bmatrix},
 \end{equation}  
Here the  functions $f^{11}, f^{12}, f^{21}$ can each be expressed in terms of $K^{\rm soft}(X,Y)$,
and $Z_{2k}$ is the elementary $2k \times 2k$ anti-symmetric matrix $\mathbb I_k \otimes \begin{bmatrix} 0 & -1 \\ 1 & 0 \end{bmatrix}$. 
 Explicitly, in the simplest case $k=1$, we have \cite[Eq.~(4.15) with $w \mapsto - w$]{Fo13}
 \begin{multline}\label{C1e} 
  \rho_{(1)}(X)   =    {1 \over 2} K^{\rm soft}(X,X) \\ - {1 \over 2} \int_{-\infty}^X e^{-w(X - t)/2} 
 {\partial \over \partial X} K^{\rm soft}(t,X)  \, dt 
  + {w \over 4}  \int_{-\infty}^X  dt \, e^{-w(X - t)/2}  \int_X^\infty du \, {\partial \over
 \partial t}  K^{\rm soft}(u,t).
  \end{multline}  
  \end{remark}

   \subsection{Characterisation of the hard edge critical regime for general $\beta > 0$}\label{S3.4}
   Analogous to the tridiagonal reduction of (\ref{1.6}) introduced in \S \ref{S2.2} and $\beta$
   generalised in \S \ref{S2.4}, the Wishart matrix (\ref{3.1b}) with correlation matrix (\ref{3.1b+})
   can, by the application of Householder transformations, be reduced to a tridiagonal form
   which allows for a $\beta$ generalisation. This follows by first  applying
 Householder transformations to reduce $\tilde{X}^\dagger$ to the $N \times N$ bidiagonal form
  \cite{Si85,DE02}
 \begin{equation}\label{Bd}
B_\beta^\dagger :=
{1 \over \sqrt{\beta} } \begin{bmatrix} \sqrt{b} \chi_{\beta n} & & &  \\
\chi_{\beta (N - 1)} & \chi_{\beta (n - 1)} & &  \\
 & \chi_{\beta (N - 2)} & \chi_{\beta (n - 2)} &  & \\
 & \ddots & \ddots &  & \\
& & \chi_\beta & \chi_{\beta (n - N + 1)} 
\end{bmatrix}
\end{equation}
with $\beta = 2$.
Here $\chi_p^2$ refers to the particular gamma distribution $\Gamma[p/2,2]$,
 and $n-N$ zero columns which do not effect the non-zero eigenvalues of $X^\dagger X$ 
have been removed.  We know from \cite{Fo13} that for general $\beta > 0$, $B_\beta^\dagger B_\beta$
has eigenvalue PDF proportional to
\begin{equation}\label{PB}
\prod_{j=1}^N \lambda_j^{\beta (n - N + 1)/2 - 1} e^{- \beta \lambda_j/2}
\prod_{1 \le j < k \le N} |\lambda_j - \lambda_k|^\beta
\int_{-\infty+ic}^{\infty  + i c} e^{i t} \prod_{j=1}^N \Big ( i t - {b -1 \over 2 b} \beta \lambda_j \Big )^{-\beta/2}
\, dt,
\end{equation}
(cf.~(\ref{A3+}))
which for $\beta = 2$ is consistent with (\ref{4.1b}) and the evaluation of the matrix integral as implied by
Remark \ref{R3.3}.

Since Wishart matrices are positive definite, the smallest eigenvalue is the eigenvalue closest to the
origin. With $\Sigma = \mathbb I$ in (\ref{3.1b}) , a well defined statistical state in the neighbourhood
of the origin --- refereed to as the hard edge --- results from scaling the eigenvalues $\lambda_j \to
\lambda_j/N$ \cite{Fo93a}. In this scaling, the eigenvalues about the origin are spaced of order unity
apart, and the Laguerre parameter $\beta(n - N + 1)/2 - 1 =: a$ (this is the exponent in the first term in (\ref{PB}))
is fixed. Introducing now the covariance matrix (\ref{3.1b+}), it was shown in \cite{DF06}
by explicit calculation of the correlation functions
for $\beta = 2$ that scaling $b = c/ N$ leads to a well defined hard edge critical regime
dependent on $c$. Subsequently these scalings applied to $B_\beta^\dagger B_\beta$ were
shown to extend the meaning of this regime to general $\beta > 0$. Moreover, with $1 - \mathcal F_{\beta,c}(x)$
denoting the cumulative distribution of the smallest scaled eigenvalue, 
and thus $ \mathcal  F_{\beta,c}(x)$ equal to the probability that the interval $(0,x)$ is free of eigenvalues,
ideas relating to the derivation of
(\ref{1.K}) as applies at the soft edge critical regime were adapted to obtain an analogous characterisation at the
hard edge \cite{RR17}.

 \begin{proposition}\label{P3.9}
 The hard edge scaled distribution $\mathcal F_{\beta,c}(x)$ satisfies the partial differential equation
\begin{equation}\label{PB1} 
- x {\partial \mathcal F \over \partial x} + {2 \over \beta} c^2 {\partial^2  \mathcal F \over \partial c^2} +
\bigg ( \Big ( {2 \over \beta}(a+2) - 1 \Big ) c - c^2 - x \bigg )  {\partial  \mathcal F \over \partial c} = 0,
\end{equation}
subject to the boundary conditions
\begin{equation}\label{PB2} 
\mathcal F_{\beta,c}(0) = 1, \quad \lim_{x \to \infty}  \mathcal F_{\beta,c}(x) = 0,   \quad \lim_{c \to 0^+} \mathcal  F_{\beta,c}(x) = 0.
\end{equation}
\end{proposition}

In the special case $a=0$, it is known from \cite[Eq.~(3.24)]{Fo13} that for finite $N$ and with
$\Sigma = {\rm diag} \, (b_1,\dots, b_N)$ that the probability of no eigenvalues in the
interval $(0,s)$ for the $\beta$ generalisation of (\ref{3.1b}) has the simple functional form
$e^{-s \sum_{j=1}^N(1/2 b_j)}$ and hence
\begin{equation}\label{PB3} 
\mathcal F_{\beta,c}(x) \Big |_{a=0} = \exp \bigg ( - {\beta x \over 2} \Big ( {1 \over c} + 1 \Big ) \bigg ).
\end{equation}
It is a simple task to show that (\ref{PB3}) is consistent with Proposition \ref{P3.9}.
Note the large $c$ limiting behaviour
\begin{equation}\label{PB4} 
\lim_{c \to \infty} \mathcal F_{\beta,c}(x) \Big |_{a=0} = E_\beta^{\rm hard}(s;a) \Big |_{a=0},
\end{equation}
where $E_\beta^{\rm hard}(s;a)$ denotes the probability of no eigenvalues in $(0,s)$ of the scaled hard 
state for $\Sigma = \mathbb I$; this relies on knowledge of the formula $E_\beta^{\rm hard}(s;a)  |_{a=0} =
e^{-\beta x/2}$ \cite{Fo94}. The formula (\ref{PB4}) is to be expected for general $a > - 1$.

Rumanov \cite{Ru14} has found a Lax pair solution of (\ref{PB1}), (\ref{PB2}) in the cases
$\beta = 2, 4$ analogous to that of Proposition \ref{P2.11} for the distribution of the soft edge
critical state. These Lax pair solutions now involve particular Painlev\'e III transcendents,
and are more complicated than for the soft edge
critical state. Nonetheless, it is shown in \cite{Ru14} that in the limit $a \to \infty$, and with suitable
scaling of $c$ and $x$ that the results for the latter can be reclaimed.
Previously Painlev\'e III transcendent evaluations were known for $E_\beta^{\rm hard}(s;a)$ in the
case $\beta = 2$ \cite{TW94b} and $\beta = 1,4$ \cite{Fo99b}.

\begin{remark}
For $\beta = 2$ the hard edge critical state is a determinantal point process.
The explicit form of the correlation kernel is given in \cite{DF06} and
\cite[\S 7.2.4]{Fo10}.
\end{remark}

\section{Rank 1 perturbations with two-dimensional support}\label{S4}
\subsection{An additive rank $1$ anti-Hermitian perturbation for the GUE}\label{S4.1}
Let $A$ be an Hermitian matrix with fixed eigenvalues $\{ \mu_j \}_{j=1}^N$ ordered
as in (\ref{1.8d}), and let $\hat{\mathbf v}$ be an $N \times 1$ column vector chosen
uniformly on the sphere in $\mathbb C^N$. Form the projection matrix $\hat{\mathbf v} \hat{\mathbf v}^\dagger$,
and use it to create the additive rank $1$ anti-Hermitian perturbation of $A$, 
\begin{equation}\label{AS1}
A + i \alpha \hat{\mathbf v} \hat{\mathbf v}^\dagger, \qquad \alpha > 0.
\end{equation}
This model, in the case of $A$ is random from the GOE, was first considered by Ullah \cite{Ul69} in the context
of resonances in scattering processes.
The working leading to (\ref{1.8c}) tells us that the eigenvalues of (\ref{AS1})
are determined by the solution of the equation, in the variable $z$,
\begin{equation}\label{AS2}
0 = 1 - i \alpha \sum_{j=1}^N { | v^{(j)} |^2 \over z - \mu_j}.
\end{equation}
As noticed in \cite{HKSSS96,SS98}, for $\alpha \to \infty$ this implies that the $N-1$
of the eigenvalues, which are in general complex, will approach the real
axis and interlace with the sequence (\ref{1.8d}); in fact
 ${\rm Re} \, z$ is between $\mu_N$ and $\mu_1$ for each solution \cite{DE21}.
 The remaining eigenvalue in this
limit can be read off from (\ref{AS2}) by searching for a solution with $|\mu|$ large.
This gives, after averaging over the components of $\mathbf v$,
\begin{equation}\label{AS3}
z \sim {i \alpha }.
\end{equation}
The large $z$ form of (\ref{AS2}) implies a sum rule constraining the eigenvalues
for fixed $\alpha$. Thus, with $\{z_l\}$ the eigenvalues, partial fractions give
that the rational function in (\ref{AS2}) can be written as $\prod_{l=1}^N (z - z_l)$
divided by 
$\prod_{l=1}^N (z - \mu_l)$. Equating the coefficient of $1/z$ in the large
$z$ expansion of both expressions and taking imaginary parts shows
\begin{equation}\label{AS3a}
\alpha = \sum_{l=1}^N {\rm Im} \, z_l.
\end{equation}

Note the similarity between (\ref{AS3}) and the $\alpha \to \infty$ form the outlier (\ref{1.6b}) for (\ref{1.6}),
as well as the corresponding interlacing in this limit. A further known general property of 
(\ref{AS2}) is that the roots $\{ \mu_j \}$ all have positive imaginary parts
\cite{OW17,Ko17,DE21}. One way to see this is, in keeping with the Schur decomposition  discussed in
the text including (\ref{4.1a}) below, to conjugate (\ref{AS1}) by a unitary matrix, bring it to triangular form with
diagonal entries equal to the eigenvalues. By inspection the diagonal entries on the RHS have positive
imaginary part for $\alpha > 0$.

Specialise now, as in the references \cite{SS98,FK99}, to the case that $A$ is random
from the GUE, which in the physics application corresponds
to a broken time reversal symmetry.
  Since such matrices are unchanged by conjugation with unitary matrices, the eigenvalues
of the perturbed matrix (\ref{AS1}) are the same as for
\begin{equation}\label{AS4}
A + i \alpha \hat{\mathbf 1} \hat{\mathbf 1}^\dagger \quad {\rm or} \quad A + i \alpha {\rm diag} \, (1,0,\dots,0).
\end{equation}
 In the case that $\alpha$ itself is a random
variable with distribution $\Gamma[N-1,\alpha_0]$, which is equivalent to replacing
$\hat{\mathbf v} $ in (\ref{AS1}) by a standard complex Gaussian vector, and setting $\alpha = \alpha_0$,
the exact form of the joint eigenvalue PDF was first calculated in \cite{SS98}.
A different working was later given by Fyodorov and Khoruzhenko \cite{FK99},
allowing the eigenvalue PDF of (\ref{AS4}) to be determined directly.

\begin{proposition}\label{P4.1}
The eigenvalue PDF of the random matrices (\ref{AS4}) in the case that $A$ is a chosen from the
GUE, with the eigenvalues denoted $\{ z_l = x_l + i y_l \}$, is proportional to
\begin{equation}\label{4.1+}
\prod_{l=1}^N e^{- (x_l^2 - y_l^2)}  \prod_{1 \le j < k \le N} | z_k - z_j |^2
{e^{ - \alpha^2} \over \alpha^{N-1}} \delta \Big ( \alpha -  \sum_{l=1}^N y_l \Big ),
\end{equation}
supported on $y_l > 0$, ($l=1,\dots,N$).
\end{proposition}

\begin{proof}
Following \cite{FK99} we adopt the viewpoint that the sum of an Hermitian matrix
$H$ and anti-Hermitian matrix $i \Gamma$ is a complex matrix $J = H + i \Gamma$.
With $H$ chosen from the GUE and $\Gamma$ fixed, the distribution on $J$
is proportional to 
\begin{equation}\label{4.1a+}
e^{- {\rm Tr} ({\rm Re} \, J)^2} \, \delta  ( \Gamma - {\rm Im} \, J ), \quad  {\rm Im} \,  J := {1 \over 2i}(J - J^\dagger), \quad  {\rm Re} \,  J := {1 \over 2}(J + J^\dagger).
\end{equation}
The next step is to write $J$ in terms of its Schur decomposition, $J = U T U^\dagger$.
Here $U$ is a unitary matrix, unique up to the phase of each column, and $T$ is an
upper triangular matrix with the elements on the diagonal the eigenvalues. The Jacobian for the change of variables
is $  \prod_{1 \le j < k \le N} | z_k - z_j |^2$; see e.g.~\cite[Eq.~(15.9)]{Fo10}.

Regarding the matrix delta function in (\ref{4.1a+}), we have the matrix integral form over 
Hermitian matrices $A$ \cite[Eq.~(3.27)]{Fo10}
\begin{equation}\label{4.1b+}
\delta( \Gamma - {\rm Im} \, J) \propto \int e^{i {\rm Tr} \, (A (\Gamma - {\rm Im} \, J))} \, (d A)  
= \int e^{-i {\rm Tr} \, ( A {\rm diag}\,(y_1,\dots,y_N))} e^{-i {\rm Im} \, {\rm Tr} \, (A \tilde{T})} e^{i {\rm Tr} \, (A U^\dagger \Gamma U)}  \, (d A),
\end{equation}
where $\tilde{T}$ refers to the strictly upper triangular portion of $T$. In obtaining the final expression, use has been
made of the invariance of the distribution of $A$ upon the mapping $A \mapsto U^\dagger A U$. We also have
\begin{equation}\label{4.1c+}
e^{- {\rm Tr} \, ((J + J^\dagger)/2)^2} = e^{- \sum_{l=1}^N x_l^2} e^{- {\rm Tr} \, (\tilde{T} \tilde{T}^\dagger )/2 }.
\end{equation}
Multiplying the RHS of (\ref{4.1c+}) with the RHS of (\ref{4.1b+}) and observing by completing the square that
\begin{equation}\label{4.1d+}
\int e^{- {\rm Tr} \, (\tilde{T} \tilde{T}^\dagger)/2 }  e^{-i {\rm Im} \, {\rm Tr} \, (A \tilde{T})}  \, (d \tilde{T}) \propto
e^{ -  {\rm Tr} \, (\tilde{A} \tilde{A}^\dagger)/2 }, 
\end{equation}
where $\tilde{A}$ denotes the strictly upper triangular portion of $A$, we are left with
$$
 \int e^{-i {\rm Tr} \, ( A {\rm diag}\,(y_1,\dots,y_N))} e^{i {\rm Tr} \, (A U^\dagger \Gamma U)} e^{ -  {\rm Tr} \, (\tilde{A} \tilde{A}^\dagger)/2 } \, (d A).
 $$
 Integrating over the independent elements of $\tilde{A}$ by completing the square gives the term $e^{-2  {\rm Tr} \, \tilde{\Gamma} \tilde{\Gamma}^\dagger  }$ in analogy with
 (\ref{4.1d+}). This leaves an integration over the diagonal entries of $A$.  Writing $K = {\rm diag} \,(a_{11},\dots,a_{NN})$, and integrating
 too over the invariant measure $[U^\dagger dU]$  of the unitary matrices $U$ in the Schur decomposition gives for the eigenvalue
 PDF, up to proportionality,
 \begin{equation}\label{4.1e+}
  e^{- \sum_{l=1}^N (x_l^2 - y_l^2)}   e^{-  {\rm Tr} \, {\Gamma}^2 }
     \prod_{1 \le j < k \le N} | z_k - z_j |^2 \, \int (dK) \, e^{- i {\rm Tr} \, K  \, {\rm diag} \, (y_1,\dots,y_N)} 
  \int [U^\dagger dU] \, e^{i {\rm Tr} \, (K U^\dagger \Gamma U)} .
 \end{equation}
 Here we have also used the fact that $e^{-2  {\rm Tr} \, \tilde{\Gamma} \tilde{\Gamma}^\dagger  } = 
e^{\sum_{l=1}^N y_l^2}   e^{-  {\rm Tr} \, {\Gamma}^2 }$; cf.~(\ref{4.1c+}).

All the above working holds for general $\Gamma$. We now specialise to the rank $1$ case as implied
by (\ref{AS4}). The matrix integral over $[U^\dagger dU]$ is then the rank $1$ HCIZ integral and
so evaluates to a single contour integral. Substituting in (\ref{4.1e+}) then shows the integrals over the diagonal
matrices $K$ factorise as a product of $N$ independent contour integrals, each of which can
be evaluated by closing the contour and computing the residue.
The final integral over $t$ is then a delta function, accounting for all terms in (\ref{4.1+}).
\end{proof}

\begin{remark}
A method of proof of (\ref{4.1+}), together with a $\beta$ generalisation has been given by Kohzan \cite{Ko17},
which is based on a finite $N$ tridiagonal formalism; see also \cite{AK21}.
\end{remark}

The bulk scaling of the eigenvalues at the origin for the GUE is specified by $x_j \mapsto X_j/\sqrt{2N}$
to give an expected density of $1/\pi$;
see e.g.~\cite[\S 7.1.1]{Fo10}. This scaling can be carried out in (\ref{AS4}) by multiplying the matrix sum by $\sqrt{2N}$,
which in turn requires  that in (\ref{4.1+}) the imaginary part of the eigenvalues be similarly scaled
$y_j \mapsto  Y_j/\sqrt{2N}$. If we further scale $\alpha \mapsto  \sqrt{N/2} \alpha_0$, then the delta function
constraint in (\ref{4.1+}) tells us that ${1 \over N} \sum_{j=1}^N Y_j = \alpha_0$, and thus on average each $Y_j$ is of
order unity. It was shown in \cite{FK99} that this scaled limit gives rise to a determinantal point process and the
explicit form of the correlation kernel was computed.

\begin{figure*}
\centering
\includegraphics[width=0.65\textwidth]{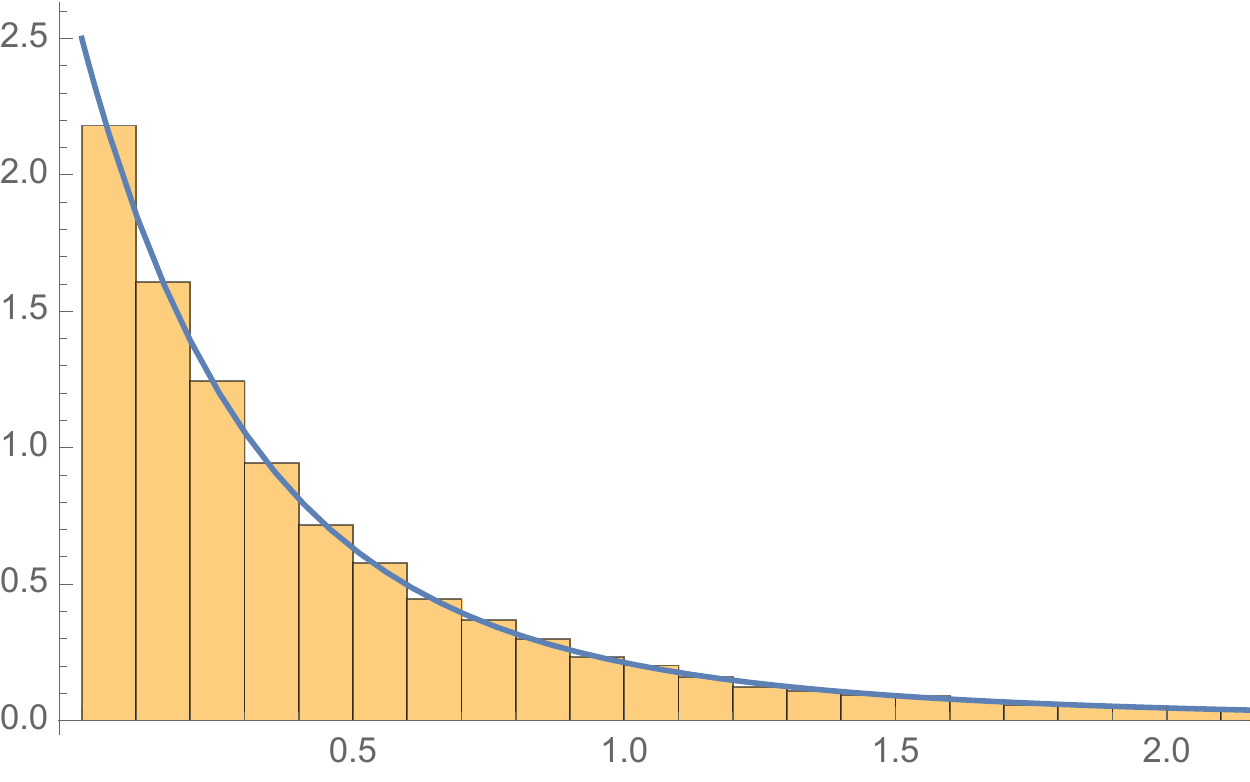}
\caption{Simulated histogram of the PDF for the scaled imaginary part of  (\ref{AS1})
 with $\alpha_0=2$ and thus $g=5/2$, plotted against (\ref{4.1h}).
 In the simulations the scaled version of (\ref{AS1}) was sampled $5,000$ times,
  and the imaginary part of the $8$ eigenvalues with real part closest to the origin were
  recorded each time.}
  \label{F4.3}. 
\end{figure*}

\begin{proposition}\label{P4.2}
Consider the scalings of (\ref{AS4}) and the corresponding eigenvalues as specified in the above 
paragraph. The correlations have the determinantal form (\ref{B0}) with correlation kernel
 \begin{equation}\label{4.1f}
 K(Z_j, Z_k) = e^{-  g (Y_j + Y_k)} \int_{-1}^1 (g + s ) e^{i s (Z_j - \bar{Z}_k)} \, ds,
 \end{equation}
 where $Z = X + i Y$ and $g = (1/2) (\alpha_0 + 1/\alpha_0)$. Here the normalisation has been chosen
 so that $\int_0^\infty K(Z,Z) \, d Y = 1$.
\end{proposition}

\begin{remark}\label{R4.3} 
1.~For finite $N$ the delta function constraint in (\ref{4.1+})  prohibits  a
 determinantal form for the correlations. The starting point of the calculation in \cite{FK99} is
 to write the $k$-point correlation function in terms of a certain product of determinants
 averaged over the GUE. \\
 2.~The reproducing property of the kernel
  \begin{equation}\label{4.1g}
  \int_{-\infty}^\infty d X_2  \int_{0}^\infty d Y_2 \,  K(Z_1, Z_2)  K(Z_2, Z_3) =  K(Z_1, Z_3),
  \end{equation}
  which is associated with perfect screening (see \cite[\S 14.1]{Fo10}), is readily verified. \\
  3.~The normalised density profile in the $Y$ direction, $\rho_{(1)}(Y)$, is obtained by setting $Z_j = Z_k = Z$
  in (\ref{4.1f}), which shows
    \begin{equation}\label{4.1h}
 \rho_{(1)}(Y)  =  e^{- 2  g Y} \Big ( g {\sinh 2 Y \over  Y} -  {\partial \over \partial Y} {\sinh 2 Y \over  2 Y} \Big ).
  \end{equation}
  In Figure \ref{F4.3} this profile is compared against a simulation for a particular $g$. 
  For the random matrices (\ref{AS4})  with $A$ chosen from the GOE, upon the scaling as used in
  Proposition \ref{P4.2} an exact evaluation of $\rho_{(1)}(Y)$ is also known \cite{SFT99},
  while unlike when $A$ chosen from the GUE, the corresponding higher order correlations remain unknown.\\
  4.~Take the viewpoint that in (\ref{AS4}) one realisation of a GUE matrix is chosen, and $\alpha \ge 0$ is 
  a continuous parameter. From the discussion at the beginning of this section we know
  that for $\alpha >0 $ the eigenvalues have a positive imaginary part, although that as $\alpha \to \infty$
  all but one eigenvalue --- which can be considered as an outlier --- returns to the real axis. With the GUE matrix scaled by dividing by $\sqrt{N/2}$ so the leading order support
  is $(-2,2)$, 
  in the recent work \cite{DE21} it has been proved that with high probability the outlier can distinguished for all $\alpha > 1 + N^{-1/3+ \epsilon}$, $\epsilon > 0$;
  see Figure \ref{F4.4} for an illustration. Moreover, it was emphasised that the exponent $-1/3$ is identical to that giving rise to the soft edge critical regime
  of the additive perturbed GOE and GUE, as displayed in the second scaling relation of (\ref{C1}).
  \end{remark}   

\begin{figure*}
\centering
\includegraphics[width=0.65\textwidth]{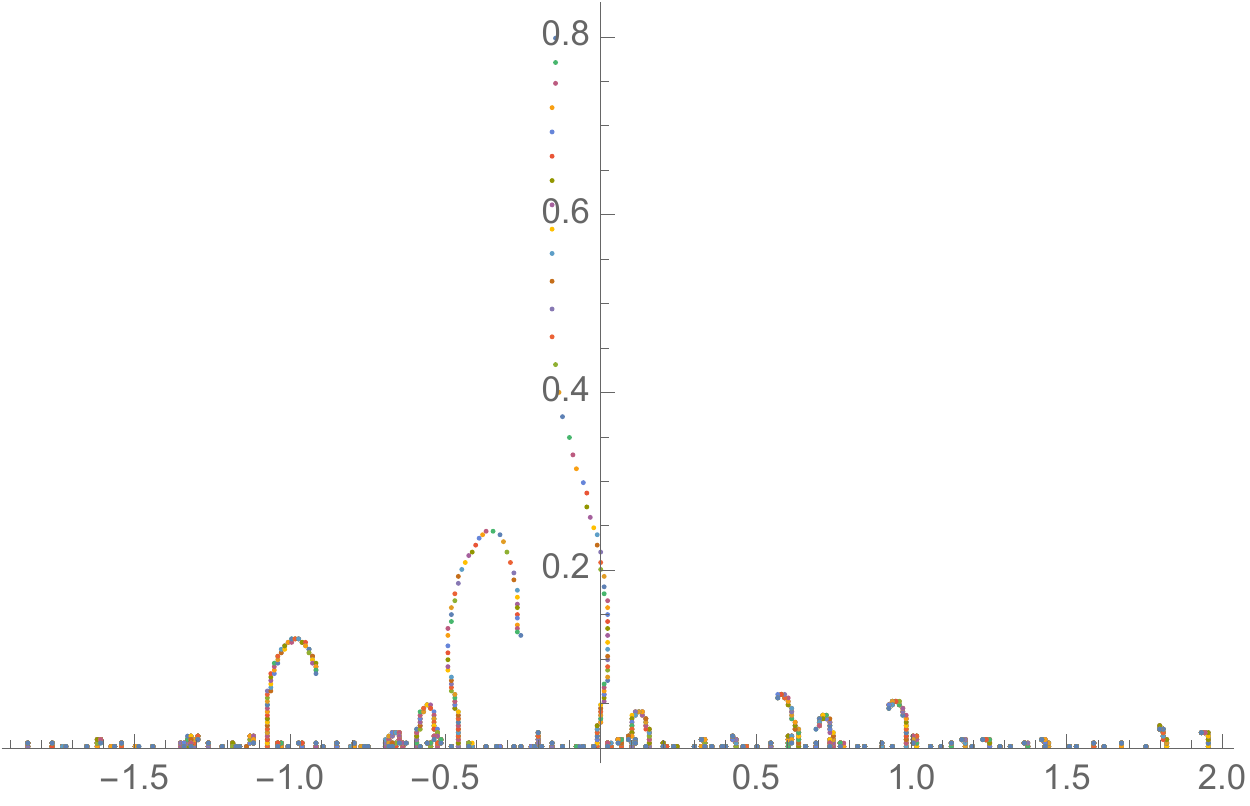}
\caption{Trajectories of the eigenvalues of the matrix  (\ref{AS4}) with one realisation of $A$ as a $100 \times 100$ scaled  GUE matrix,
and $\alpha$ varying from $0$ to 1.5 in intervals of $1/60$. The outlier is clearly visible.}
  \label{F4.4}. 
\end{figure*}

\subsection{A multiplicative sub-unitary rank $1$ perturbation for the CUE}\label{S4.2}
Closely related to the   anti-Hermitian additive rank $1$ perturbation for the GUE of the previous section
is a particular multiplicative rank $1$ perturbation of CUE matrices. Here CUE stands for the circular
unitary ensemble, this being terminology introduced by Dyson \cite{Dy62}, which is realised by
the set of complex unitary matrices distributed according to Haar measure. Let $A = {\rm diag} \, (a,1,\dots,1)$,
with $|a| < 1$, and consider the multiplicative rank $1$ perturbation of $U \in \, {\rm CUE}$ define by
$UA$. The joint eigenvalue PDF was shown by Fyodorov 
\cite{Fy01} (see also the review \cite{FS03}), to be proportional to
\begin{equation}\label{q0}
(1 - |a|^2)^{1-N}  \delta \Big ( |a|^2 - \prod_{l=1}^N | z_l|^2 \Big )
\prod_{1 \le j < k \le N} | z_j - z_k|^2,
\end{equation}
and supported on $|z_l|<1$.

As for the PDF (\ref{4.1+}), the delta function constraint prohibits a determinantal form of the correlations
for finite $N$.
Nonetheless, in distinction to (\ref{4.1+}), an exact finite $N$ expression is still possible \cite{Fy01}.
To present this result, define
\begin{equation}\label{q}
q_j(z_1,\dots,z_k) = [s^j] \det \Big [ \Big (s + x {d \over dx} \Big ) {x^N - 1 \over x - 1} \Big |_{x = z_i \bar{z}_j}
\Big ]_{i,j=1}^k,
\end{equation}
where $[s^j]$ denotes the coefficient of $s^j$ in the expression that follows.

\begin{proposition}\label{P5.1}
Require that $|z_l| < 1$, ($l=1,\dots,k)$ and $\prod_{l=1}^k |z_l|^2 \ge |a|^2$ we have 
\begin{equation}\label{12}
\rho_{(k)}(z_1,\dots,z_k) =  {1 \over \pi^k} (1 - |a|^2)^{1-N}\
\sum_{l=0}^k q_l(z_1,\dots,z_k) \Big ( {d \over dx} x \Big )^l
\Big [ {1 \over x} \Big ( 1 - {|a|^2 \over x} \Big )^{N-1} \Big ] \Big |_{x = \prod_{l=1}^k |z_l|^2}.
\end{equation}
\end{proposition}

\begin{remark}
1.~Considering 
 the case $a=0$. Then the only nonzero term in 
(\ref{12}) is  $l=0$, implying
\begin{equation}\label{qA}
\rho_{(k)}(z_1,\dots,z_k) \Big |_{a=0} = {1 \over \pi^k} 
\det \Big [  {d \over dx} {x^N - 1 \over x - 1} \Big |_{x = z_i \bar{z}_j} \Big ]_{i,j=1}^k.
\end{equation}
On the other hand setting $a=0$ in the definition of $A$ and forming $UA$ shows the
resulting matrix has the first column of $A$ replaced by a column of zeros. Hence
there is one zero eigenvalue, with the remaining eigenvalues the same as those
of the $(N-1) \times (N-1)$ submatrix of $U$ obtained by deleting the first row and the
first column. For this ensemble the joint eigenvalue PDF was first derived in
 \cite{ZS99} to be proportional to $\prod_{1 \le j < k \le N-1} | z_j - z_k|^2$, supported
 on $|z_l| < 1$. Without any delta function constraint, this corresponds to a determinantal
 point process, and the $k$-point correlation (\ref{qA}) was obtained in  \cite{ZS99}.  \\
 2.~The density formula $k=1$ of (\ref{12}) is a special case of a formula for the eigenvalue
 density of $U \sqrt{G}$, $G = {\rm diag} \, (g_1,\dots, g_N)$, each $g_i \ge 0$, obtained
 in \cite{WF08}. \\
 3.~In analogy with Remark \ref{R4.3} point 4.~varying the parameter $a$ in $A$ from $1$ to $0$
 with $U$ a single sample
 gives rise to an eigenvalue process where all eigenvalues begin on the unit circle for
 $a=1$, and as $a$ varies to $0$ exactly one eigenvalue ends at $z=0$. However unlike the
 setting for the random matrices (\ref{AS4}), the remaining eigenvalues do not return
 to the unit circle, although in a qualitative sense most do remain close to the unit circle;
 see Figure \ref{F4.5} for an illustration.
  \end{remark}
  
 \begin{figure*}
\centering
\includegraphics[width=0.45\textwidth]{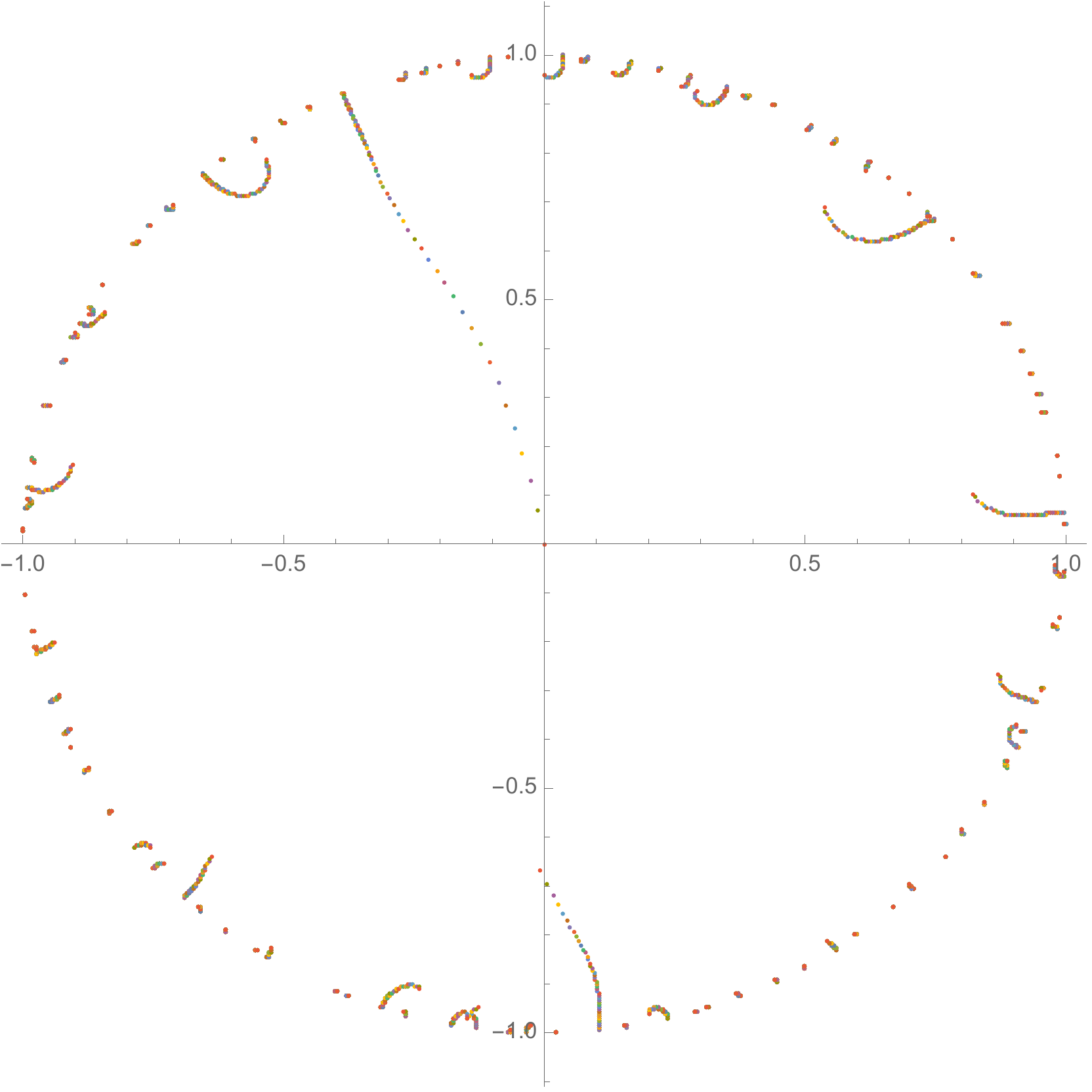}
\caption{Trajectories of the eigenvalues of the matrix  $UA$ with one realisation of $U$ as a $100 \times 100$ scaled  CUE matrix,
and $a$ varying from $1$ to $0$ in intervals of $0.01$. }
  \label{F4.5}. 
\end{figure*}

The finite $N$ result (\ref{12}) admits two distinct scaling limits. One is to expand the coordinates in the neighbourhood of the
boundary of the unit circle. This can be done by writing $z_j = (1 + (i x_j - y_j)/N + O(1/N^2))$. Now taking $N \to \infty$ reduces
(\ref{12}) to the form (\ref{B0}) for a determinantal point process \cite{Fy01}. In fact this form is precisely 
(\ref{4.1f}) as found in relation to the bulk scaling limit of (\ref{AS4}), with the identification $g = 2/(1-|a|^2) - 1$. The second scaling
limit scales the coupling $a$ but leaves the coordinates in the disk unchanged \cite{FI19}. Then the finite $N$ structure
(\ref{12}) is conserved, and thus the statistical state is not a determinantal point process.

\begin{proposition}\label{P2}
Consider (\ref{12}) with the substitution (\ref{q}). Scale the parameter $a$ to depend on $N$ according to
$a = 1/(\mu \sqrt{N})$, and define
\begin{equation}\label{Q}
Q_j(z_1,\dots,z_k) = [s^j] \det \Big [ \Big (s + x {d \over dx} \Big ) {1 \over 1 - x} \Big |_{x = z_i \bar{z}_l}
\Big ]_{i,l=1}^k
\end{equation}
(cf.~(\ref{q})).
In the limit $N \to \infty$ the general $k$-point correlation is specified by
\begin{equation}\label{R2a}
\rho_{(k)}(z_1,\dots,z_k) = { e^{1 / |\mu|^2}  \over \pi^k} 
\sum_{l=0}^k Q_l(z_1,\dots,z_k) \Big ( {d \over dx} x \Big )^l
\Big ( {1 \over x} e^{- 1/ ( |\mu|^2 x )} \Big ) \Big |_{x = \prod_{l=1}^k |z_l|^2}.
\end{equation}
restricted to $|z_j|<1$ ($j=1,\dots,k$).
\end{proposition}

\begin{remark}
Setting $k=1$ in (\ref{12}) and simplifying shows
\begin{equation}\label{R1}
\rho_{(1)}(z) = {1 \over \pi} {1 \over (1 - |z|^2)^2} \exp \Big ( {1 \over  |\mu|^2} \Big ( 1- {1 \over |z|^2} \Big ) \Big )
\Big (1 +{1 \over  |\mu|^2  |z|^4} (1 - |z|^2) \Big ).
\end{equation}
This exhibits an essential singularity as $|z| \to 0^+$.
  \end{remark}

\subsection{Relationship to the limiting Kac polynomial}
In the theory of random polynomials the Kac polynomial refers to the $N$-th degree
polynomial $\sum_{n=0}^N c_n z^n$, where each coefficient $c_n$ is an independent
standard real Gaussian \cite{Ka43}. Subsequently the complex version, where the
vector of coefficients is specified to have a vector complex Gaussian distribution with
general covariance matrix, was considered by Hammersley \cite{Ha56}. Some years later
special properties of this complex version, in the Kac
setting where each coefficient is independent and identically distributed with mean zero,
and extended to an analytic function by taking $N \to \infty$,
where identified  \cite{PV03,Kr06}. 

First, in \cite{PV03} it was shown in this setting
the statistical state is a determinantal point process, with the general $k$-point correlation
supported on $|z_l| < 1$ and correlation kernel
\begin{equation}\label{qAc}
\rho_{(k)}(z_1,\dots,z_k) = {1 \over \pi^k} 
\det \Big [   {1 \over 1 - z_i \bar{z}_j } \Big ]_{i,j=1}^k.
\end{equation}
Comparison of (\ref{qAc}) with (\ref{qA}) shows that the latter, in the limit $N \to \infty$, coincides
with the former. In fact in \cite{Kr06} it was proved directly that the characteristic polynomial for the
eigenvalues of a random matrix from the CUE, with one row and one column deleted is for $N \to
\infty$ given by $\sum_{n=0}^\infty c_n z^n$, where each coefficient $c_n$ is an independent
standard complex Gaussian. There is a simple extension of this result in relation to the
characteristic polynomial of $UA$ as considered in the previous subsection,
with the parameter $A$ scaled as in Proposition \ref{P2} \cite{FI19}.

 \begin{proposition}\label{P4.7}
 Consider the random matrix $UA$ as specified in the previous subsection, and set $a = 1/ (\mu \sqrt{N})$.
 In the limit $N \to \infty$ the eigenvalues of $UA$ are given by the zeros of the random Laurent series 
 \begin{equation}\label{L1a}
 {1\over \mu} - \sum_{j=1}^\infty {c_j \over z^j}
 \end{equation} 
 in the variable $\lambda = 1/z$, $|\lambda| < 1$, 
 with each $c_j$ an independent standard
 complex Gaussian.
 \end{proposition}
 
 \begin{proof} (Sketch)
 The first step is to manipulate the characteristic polynomial for $UA$ to conclude that the condition
 for an eigenvalue $\lambda$ of $UA$ can be written
 \begin{align}\label{W}
0   & =    \det \Big ( \mathbb I_N - {\lambda \over a} ( \mathbb I_N - \lambda U^\dagger \mathbb I_N')^{-1} U^\dagger \hat{\mathbf{e}}_N^{(1)} (\hat{\mathbf{e}}_N^{(1)})^T \Big )  \nonumber \\
& = 1 - {\lambda \over a}    (\hat{\mathbf{e}}_N^{(1)})^T  (  \mathbb I_N -   \lambda U^\dagger \mathbb I_N')^{-1}   U^\dagger \hat{\mathbf{e}}_N^{(1)} \nonumber  \\
& =  1 - {\lambda \over a}  \sum_{k=0}^\infty \lambda^k  (\hat{\mathbf{e}}_N^{(1)})^T  (U^\dagger I_N')^k U^\dagger \hat{\mathbf{e}}_N^{(1)}.
\end{align}
Here $\mathbb I_N'$ denotes the identity matrix with the first diagonal replaced by $0$,  the second line is obtained by using the determinant identity
(\ref{1.8b}), and the final line is obtained by using the geometric series to expand the matrix inverse, which is valid for $|\lambda| < 1$.

 It was established in  \cite{Kr06} that for $V \in U(N)$ chosen with Haar measure, in the limit $N \to \infty$
  \begin{equation}\label{Va} 
  \sqrt{N} ( (\mathbf{e}_N^{(1)})^T V \mathbf{e}_N^{(1)}, (\mathbf{e}_N^{(1)})^T V^2 \mathbf{e}_N^{(1)},  (\mathbf{e}_N^{(1)})^T V^3 \mathbf{e}_N^{(1)}, \dots )
  \mathop{=}^{\rm d} (\alpha_1,\alpha_2,\alpha_3,\dots ),
  \end{equation}
  where each $\alpha_i$ is an independent standard complex Gaussian. The stated result now follows
 by noting that for $a = 1/ (\mu \sqrt{N})$,
 \begin{equation}\label{Va1} 
(\mathbf{e}_N^{(1)})^T (U^\dagger I_N')^k U^\dagger \mathbf{e}_N^{(1)} =
( \mathbf{e}_N^{(1)})^T (U^\dagger)^{k+1}  \mathbf{e}_N^{(1)} \Big ( 1 + {\rm O}(k/N) \Big ).
  \end{equation}
\end{proof}

\begin{remark}
1.~It has been noted in the Introduction that for a real matrix $\tilde{X}$ satisfying the circular law, the real symmetric
perturbation (\ref{1.4}) results in a single outlier for $\alpha > 1$. This conclusion holds true for $\tilde{X}$ complex
and (\ref{1.4}) replaced by $\tilde{X} + \alpha \hat{\mathbf v} \hat{\mathbf v}^\dagger$, for $ \hat{\mathbf v}$ a  complex
unit vector \cite{Ta12}. On the other hand, considering instead the particular non-Hermitian rank $1$ perturbation
 \begin{equation}\label{Xv}
 \tilde{X} + \alpha \hat{\mathbf 1}_N \hat{\mathbf v}^\dagger,
  \end{equation}
  with $\hat{\mathbf v}$ chosen randomly, the situation is very different --- an important structural feature here is
  that the rank $1$ term averages to zero, in contrast to the case of this term equalling $\alpha \hat{\mathbf v} \hat{\mathbf v}^\dagger$.
   Set $\alpha = \mu/\sqrt{N}$ and take the limit $N \to \infty$.
  It was proved in \cite{Ta12}  that the eigenvalues of (\ref{Xv}) are given by the zeros of (\ref{L1a})
  with respect to the variable $z$, $|z| > 1$. Thus according to Proposition \ref{P4.7} and recalling
  that $\lambda = 1/z$, this characterisation is identical to that for the eigenvalues of the scaled random matrix $UA$.
  Note from steps analogous to the derivation of (\ref{W}) that this would follow from
  the characteristic equation for (\ref{Xv}) if it could be established that (\ref{Va}) holds true with each $(\mathbf{e}_N^{(1)})^T V^j \mathbf{e}_N^{(1)}$
  replaced by $(\mathbf{e}_N^{(1)})^T \tilde{X}^j \mathbf{v}$. This is precisely what is established in \cite{Ta12}. \\
  2.~Consider the random matrix $RA$ with $R \in O(N)$ chosen with Haar measure and $A = {\rm diag} \, (a,1,\dots, 1)$, $|a| < 1$.
  The proof of Proposition \ref{P4.7} can be modified to lead to the conclusion that with $a = 1/(\mu \sqrt{N})$, the eigenvalues are
  given by the zeros of (\ref{L1a}) in the variable $\lambda = 1/z$, $|\lambda| < 1$, with each $c_j$ an independent standard
  real Gaussian. On the other hand, consider the random matrix $\tilde{X} + \alpha \hat{\mathbf{1}}_N \hat{\mathbf v}^T$, where
  $\tilde{X}$ is a real matrix obeying the circular law and $\hat{\mathbf v}$ is a random real unit vector. For $\alpha = \mu/\sqrt{N}$ and
  $N \to \infty$ it is proved in \cite{Ta12} that this same limiting Laurent polynomial --- which is equivalent to the original Kac random
  polynomial as defined at the beginning of this subsection --- specifies the eigenvalue distribution in the region $|z| > 1$. \\
  3.~The eigenvalues of the case $a=0$ of the random matrix $RA$ --- this corresponding  to deleting one row and column of $R$ ---
  are known to form a Pfaffian point process \cite{KSZ09,Fo10a}. This point process can be considered to consist of two species, the
  real eigenvalues, and the complex eigenvalues. Statistics associated with the number of real eigenvalues
  \cite{FK18} have been shown recently to relate to the persistence exponent for two-dimensional diffusion
  with random initial conditions \cite{PS18}.
    \end{remark}
    
    \subsection{Left and right eigenvector statistics}
    From matrix theory we know that eigenvectors of a matrix $X$ form an orthonormal set iff $[X,X^\dagger]=0$.
    If $X$ is random, this will not be the case unless $X$ is Hermitian. Instead, for non-normal matrices orthonormality only
    shows itself when considering both the eigenvectors $\{ | R_j \rangle \}_{j=1}^N$ of $X$ and the eigenvectors
  $\{ | L_j \rangle \}_{j=1}^N$  of $X^\dagger$. These are referred to as the right (R) and left (L) eigenvectors respectively,
  with the distinction conveniently indicated symbolically in the bra-ket notation. Thus the left and right eigenvectors
  can be chosen so that $\langle L_i | R_j \rangle = \delta_{i,j}$, which in words says that they form a bi-orthogonal
  family. Note that with respect to this condition,
   $ |R_j \rangle $ can be  multiplied by the scalar $\alpha$ provided $|L_j \rangle $ is multiplied by
  $1/\bar{\alpha}$. Invariant under such scaling is the so-called overlaps $\mathcal O_{mn} = \langle L_m | L_n \rangle 
  \langle R_m | R_n \rangle$, with the diagonal overlaps $ \mathcal O_{nn} = || L_n||^2 || R_n ||^2$ of particular importance
  for their role as squared eigenvalue condition numbers; see \cite{Du21} and references therein.
  The overlaps $\mathcal O_{mn}$ were studied in the context of the additive rank $1$ anti-Hermitian GUE perturbation 
  (\ref{AS1}) by Fyodorov and Mehlig \cite{FM02}. 
  Very recent work of  Fyodorov and Osman \cite{FO21a,FO21b} has advanced knowledge of the diagonal overlaps
  $ \mathcal O_{nn} $ from the evaluation of the their limiting mean in the sense of the quantity $O(z)$ defined in
  (\ref{O12}) below, to the evaluation of their limiting distribution $\mathcal P(t;z)$ as specified by (\ref{4.33a}).
  Moreover, this quantity was also calculated for the GOE analogue of  (\ref{AS1}).
  In another recent development,  the overlaps  $ \mathcal O_{nm},  $, together with certain generalisations referred to
  as $q$-overlaps, were studied for the multiplicative sub-unitary rank $1$ CUE perturbation of subsection \ref{S4.2} in the case
  $a=0$ by Dubach \cite{Du21}.
  
  Here we will focus attention on results from \cite{FM02}. First, following \cite{FS97}, consider the scalar
   \begin{equation}\label{Sc}
   s = s(E) := {1 + i K \over 1 - i K} =1 + 2 {i K \over 1 - i K}, \qquad K = \alpha \mathbf v^\dagger (E \mathbb I_N - A)^{-1} \mathbf v.
   \end{equation}
   Here  $E$ is an in general complex parameter, $\alpha$ is a positive real scalar, $A$ is an $N \times N$ Hermitian matrix and
   $\mathbf v$ is an $N \times 1$ complex column vector. Hence, with $\tilde{K}$ the $N \times N$ matrix 
   $\tilde{K} = \alpha (E \mathbb I_N - A)^{-1} \mathbf v  \mathbf v^\dagger$,
   we can expand
   \begin{multline}\label{Sc1}
   {iK \over 1 - i K} =  \sum_{p=0}^\infty ( i K)^{p+1} =  i \alpha \mathbf v^\dagger \Big ( \sum_{p=0}^\infty ( i \tilde{K})^p \Big )
   (E \mathbb I_N - A)^{-1} \mathbf v \\
   =  i \alpha \mathbf v^\dagger (\mathbb I_N - i \tilde{K})^{-1}  (E \mathbb I_N - A)^{-1} \mathbf v
  =
   i  \alpha \mathbf v^\dagger ( E \mathbb I_N - A - i \alpha \mathbf v  \mathbf v^\dagger )^{-1} \mathbf v .
   \end{multline} 
   It follows from (\ref{Sc1}) that the poles of (\ref{Sc}) as a function of $E$ occurs at the eigenvalues $\{z_j \}$ of
   $\mathcal H = A + i \alpha \mathbf v \mathbf v^\dagger$, which is the matrix
   (\ref{AS1}) with $\hat{ \mathbf v}$ replaced by $t{ \mathbf v}$. Moreover, $s(E)$ has unit modulus for $E$ real and so
   has the rational function form
    \begin{equation}\label{Sc2} 
    s(E) = \prod_{j=1}^N { E - \bar{z}_j \over E -  {z}_j }.
    \end{equation}
    Making use of this allows a formula for the overlaps  $\mathcal O_{mn}$ of $\mathcal H$ to be computed \cite{FM02,FS03}.
    
    \begin{proposition}
    Consider the non-Hermitian matrix $\mathcal H$ defined in the above paragraph. Choose the left and right
    eigenvectors to form an orthonormal set. The overlaps of these eigenvectors are given in terms of the
    eigenvalues by
   \begin{equation}\label{Sc3}   
  \mathcal O_{mn} =  { (z_n - \bar{z}_n) (z_m - \bar{z}_m) \over (z_n - \bar{z}_m)^2} \prod_{k=1 \atop \ne n}^N {z_n - \bar{z}_k \over
  z_n - {z}_k} \prod_{k=1 \atop \ne m}^N {\bar{z}_m -  {z}_k \over
  \bar{z}_n - \bar{z}_k }.
 \end{equation}
 \end{proposition}  
    
 \begin{proof}
 Let the matrix of eigenvectors of $\mathcal H$ be denoted $P$, so that $\mathcal H = P Z P^{-1}$, where $Z$
 is the diagonal matrix of eigenvalues. The overlaps can be expressed in terms of $V$ according to
  $$
  \mathcal O_{mn} = (P^\dagger P)_{mn} (P^{-1}(P^{-1})^\dagger)_{nm},
 $$
 where on the RHS the subscripts indicate the positions in the corresponding  matrix.
 On the other hand, we observe
  \begin{align*}
  ( P^\dagger (\mathcal H - \mathcal H^\dagger) P )_{mn} & = (z_n - \bar{z}_m) (P^\dagger P)_{mn}  \\
  ( P^{-1} ( \mathcal H^\dagger  - \mathcal H ) (P^{-1})^\dagger)_{nm} & = (\bar{z}_m - z_n) ( P^{-1} (P^{-1})^\dagger )_{nm}.
  \end{align*}
  Hence 
  \begin{equation}\label{Sc6}  
  \mathcal O_{mn} = {4 \alpha^2 \over (z_n - \bar{z}_m) ( \bar{z}_m - z_n)}   (P^\dagger \mathbf v \mathbf v^\dagger P)_{mn} (P^{-1}  \mathbf v \mathbf v^\dagger (P^{-1})^\dagger)_{mn}.
  \end{equation} 
  
  It follows from (\ref{Sc}) and (\ref{Sc1}) that
  $$
  s(E) = 1 + 2 i \alpha \mathbf v^\dagger P (E \mathbb I_N - Z)^{-1} P^{-1} \mathbf v.
  $$
  Consider now the leading term for $E \to z_n$, which effectively replaces $(E \mathbb I_N - z)^{-1}$ by $(\mathbf e_n \mathbf e_n^\dagger) (E - z_n)^{-1}$, where $\mathbf e_n$ is the
  $n$-th standard basis vector in $\mathbb R^N$. We then have
  $$
  s(E) \sim {2 i \alpha \over E - z_n} (\mathbf v^\dagger P)_n (P^{-1} \mathbf v)_n.
  $$
  Taking the complex conjugate transpose of both sides shows that for $z \to z_m$
  $$
  \overline{s(E)} \sim  - {2 i \alpha \over \bar{E} - \bar{z}_m}    (P^{\dagger} \mathbf v)_m  
  (P^{\dagger} \mathbf v)_m.
  $$
  Multiplying together these latter two equations shows that for $E_1 \to z_n$ and $E_2 \to z_m$ we have
  $$
  s(E_1)  \overline{s(E_2)} \sim {4 \alpha^2 \over (E_1 - z_n) (\bar{E}_2 - \bar{z}_m)} (P^\dagger \mathbf v \mathbf v^\dagger P)_{nm}
  (P^{-1} \mathbf v  \mathbf v^\dagger (P^{-1})^\dagger)_{nm}.
  $$
  Making use of (\ref{Sc2}) allows this asymptotic equality to be turned into an identity, with the LHS expressed as a product over the eigenvalues.
  Substituting the resulting formula in (\ref{Sc6}) gives (\ref{Sc3}).
   \end{proof} 
   
   Associated with the overlaps $ \mathcal O_{mn} $ are ensemble averages
   \begin{equation}\label{O12}
   O(z) = \bigg \langle {1 \over N} \sum_n \mathcal O_{nn} \delta (z - z_n)   \bigg \rangle, \quad
   O(z,z') = \bigg \langle {1 \over N} \sum_{m \ne n}  \mathcal O_{mn} \delta (z - z_m)  \delta (z - z_n)   \bigg  \rangle.
 \end{equation} 
 For the class of non-Hermitian rank 1  perturbations (\ref{AS4}), with $N$ fixed these are given by substituting (\ref{Sc3}) for $ \mathcal O_{nn},  \mathcal O_{mn}$ and
 integrating against the functional form (\ref{4.1+}) for the eigenvalue PDF. With the scalings as specified in the paragraph above
 Proposition \ref{P4.2}, the large $N$ forms of both the averages in (\ref{O12}) were computed by Fyodorov and Mehlig \cite{FM02}.
 With $Y = {\rm Im} \, z$ we record the expression for the limiting form of $ O(z) $,
   \begin{equation}\label{O13} 
   O(Y) = e^{-4 g Y} {d \over d Y} \bigg ( e^{2 g Y} {\sinh 2 Y \over 2 Y} \bigg );
 \end{equation}    
cf.~(\ref{4.1h}) and note in particular the normalisation $O(Y)|_{Y=0} = \rho_{(1)}(Y)|_{Y=0}$.
Although we don't present an example, analogous to Figure \ref{F4.3} this functional form for particular
parameters can be compared against numerical simulations; see \cite[Fig.~1]{FM02}.  Thus from 
(\ref{Sc3}) we have
$$
 \mathcal O_{nn} = \prod_{k=1 \atop k \ne n}^N \bigg | {z_n - \bar{z}_k \over z_n - z_k} \bigg |^2,
 $$
which when calculated at the eigenvalues used to generate the histrogram in   Figure \ref{F4.3}
and averaged within each bin gives rise to a scale factor which is to multiply the existing heights.

Generalising $O(z)$ in (\ref{O12}) is the distribution function
   \begin{equation}\label{4.33a} 
   \mathcal P(t;z) :=  \bigg \langle {1 \over N} \sum_n  \delta ( \mathcal O_{nn} - 1 - t )\delta (z - z_n)   \bigg \rangle.
  \end{equation}    
  As commented in the introductory paragraph to this subsection, the scaled limiting form of this 
  quantity, $\mathcal P_Y^{(\beta)}(t)$ say, has recently been calculated for both the ensemble   (\ref{AS1}) (the
  case $\beta = 2$) and its GOE
  analogue (the case $\beta = 1$) \cite{FO21a}. In particular, it was shown
   \begin{equation}\label{4.33b}  
  \mathcal P_Y^{(2)}(t)  = {16 \over t^3} e^{-2gY} \mathbb L_2  e^{-2gY(1 + 2/t)} I_0 \Big ( {4 Y \over t}
  \sqrt{(g^2 - 1) (1 + t)} \Big ),
  \end{equation} 
  where $I_0(z)$ is a modified Bessel function and
  $  \mathbb L_2$ is the differential operator acting on smooth functions $f(y)$ according to
   \begin{equation}\label{4.33c}  
  \mathbb L_2 f(Y) = \bigg ( 1 + \Big (    {\sinh 2 Y \over 2 Y} \Big )^2 + {1 \over 2 Y}
  \Big ( 1 -  {\sinh 4 Y \over 4 Y} \Big ) {d \over d Y} +
  {1 \over 4} \Big (  \Big (    {\sinh 2 Y \over 2 Y} \Big )^2 - 1 \Big ) {d^2 \over d Y^2} \bigg ) Y^2 f(Y).
    \end{equation} 
    As emphasised in  \cite{FO21a}, a noteworthy feature of (\ref{4.33b}) is heavy tail decay, specifically like
    $1/t^3$ as $t \to \infty$. This  implies that all the moments $\langle  \mathcal O_{nn}^k \rangle$ diverge for $k \ge 2$,
   as is known for the analogous quantity in the case of Ginibre type ensembles
   \cite{BD20,Fy18,Du21a}. Moreover computing the moment for $k=1$ reclaims (\ref{O13}).

 \subsection*{Acknowledgements}
	This research is part of the program of study supported
	by the Australian Research Council Centre of Excellence ACEMS
	and the Discovery Project grant DP210102887.
	Helpful feedback on the first draft of this work by
	Y.~Fyodorov and J.~Ipsen is most appreciated.

\nopagebreak

\providecommand{\bysame}{\leavevmode\hbox to3em{\hrulefill}\thinspace}
\providecommand{\MR}{\relax\ifhmode\unskip\space\fi MR }
\providecommand{\MRhref}[2]{%
  \href{http://www.ams.org/mathscinet-getitem?mr=#1}{#2}
}
\providecommand{\href}[2]{#2}

\end{document}